%% file: pub-arxiv.tex
\documentclass[a4paper,UKenglish]{lipics}

\usepackage{enumerate}
\usepackage{amsmath,accents,amssymb,amsfonts,xspace}
\usepackage{thmtools,thm-restate}
\usepackage{pub}

\usepackage{xcolor,pgf,tikz}
\usetikzlibrary{arrows,automata,shapes,decorations,calc,positioning}

\newtheorem{problem}{Problem}

\title{Games on graphs with a public signal monitoring\footnote{This
    work has been supported by ERC project EQualIS (FP7-308087).}}
    
\author{Patricia Bouyer}

\affil{LSV, CNRS, ENS Paris-Saclay, Universit\'e Paris-Saclay, France}

\Copyright{Patricia Bouyer}
\subjclass{}
\keywords{}

\begin{document}

\maketitle

\begin{abstract}
  We study pure Nash equilibria in games on graphs with an imperfect
  monitoring based on a public signal. In such games, deviations and
  players responsible for those deviations can be hard to detect and
  track.  We propose a generic epistemic game abstraction, which
  conveniently allows to represent the knowledge of the players about
  these deviations, and give a characterization of Nash equilibria in
  terms of winning strategies in the abstraction. We then use the
  abstraction to develop algorithms for some payoff functions.
\end{abstract}

\section{Introduction}

Multiplayer concurrent games over graphs allow to model rich
interactions between players. Those games are played as follows.  In a
state, each player chooses privately and independently an action,
defining globally a move (one action per player); the next state of
the game is then defined as the successor (on the graph) of the
current state using that move; players continue playing from that new
state, and form a(n infinite) play. Each player then gets a reward
given by a payoff function (one function per player). In particular,
objectives of the players may not be contradictory: those games are
non-zero-sum games, contrary to two-player games used for controller
or reactive synthesis~\cite{thomas02,henzinger05}.

The problem of distributed synthesis~\cite{PR79} can be formulated
using multiplayer concurrent games. In this setting, there is a global
objective $\Phi$, and one particular player called Nature. The
question then is whether the so-called grand coalition (all players
except Nature) can enforce $\Phi$, whatever Nature does. While the
players (except Nature) cooperate (and can initially coordinate),
their choice of actions (or strategy) can only depend on what they see
from the play so far. When modelling distributed synthesis as
concurrent games, information players receive is given via a partial
observation function of the states of the game. When the players have
perfect monitoring of the play, the distributed synthesis problem
reduces to a standard two-player zero-sum game. Distributed synthesis
is a fairly hot topic, both using the formalization via concurrent
games we have already described and using the formalization via an
architecture of processes~\cite{PR90}. The most general decidability
results in the concurrent game setting are under the assumption of
hierarchical observation~\cite{MW05,BMV17} (information received by
the players is ordered) or more recently under recurring common
knowledge~\cite{BM17}.

While distributed synthesis involves several players, this remains
nevertheless a zero-sum question.  Using solution concepts borrowed
from game theory, one can go a bit further in describing the
interactions between the players, and in particular in describing
rational behaviours of selfish players. One of the most basic solution
concepts is that of Nash equilibria~\cite{nash50}. A Nash equilibrium
is a strategy profile where no player can improve her payoff by
unilaterally changing her strategy. The outcome of a Nash equilibrium
can therefore be seen as a rational behaviour of the system.
While very much studied by game theoretists (e.g. over matrix games),
such a concept (and variants thereof) has been only rather recently
studied over games on graphs. Probably the first works in that
direction are~\cite{CMJ04,CHJ06,ummels06,ummels08}. Several series of
works have followed. To roughly give an idea of the existing results,
pure Nash equilibria always exist in turn-based games for
$\omega$-regular objectives~\cite{UW11} but not in concurrent games;
they can nevertheless be computed for large classes of
objectives~\cite{UW11,BBMU15,brenguier16}. The problem becomes harder
with mixed (that is, stochastic) Nash equilibria, for which we often
cannot decide the existence~\cite{UW11a,BMS14}.

Computing Nash equilibria requires to (i) find a good behaviour of the
system; (ii) detect deviations from that behaviour, and identify
deviating players (called deviators); (iii) punish them. This simple
characterization of Nash equilibria is made explicit in~\cite{CFGR16}.
Variants of Nash equilibria require slightly different ingredients,
but they are mostly of a similar vein.

In (almost) all these works though, perfect monitoring is implicitly
assumed: in all cases, players get full information on the states
which are visited; a slight imperfect monitoring is assumed in some
works on concurrent games (like~\cite{BBMU15}), where actions which
have been selected are not made available to all the players (we speak
of hidden actions). This can yield some uncertainties for detecting
deviators but not on states the game can be in, which is rather
limited and can actually be handled.

In this work, we integrate imperfect monitoring into the problem of
deciding the existence of pure Nash equilibria and computing
witnesses. We choose to model imperfect monitoring via the notion of
signal, which, given a joint decision of the players together with the
next state the play will be in, gives some information to the
players. To take further decisions, players get information from the
signals they received, and have perfect recall about the past (their
own actions and the signals they received).  We believe this is a
meaningful framework. Let us give an example of a wireless network in
which several devices try to send data: each device can modulate its
transmission power, in order to maximise its bandwidth and reduce
energy consumption as much as possible. However there might be a
degradation of the bandwidth due to other devices, and the
satisfaction of each device is measured as a compromise between energy
consumption and allocated bandwidth, and is given by a quantitative
payoff function.\footnote{This can be expressed by
  $\payoff_{\text{player}\ i} = \frac{R}{\mathsf{power}_i} \Big(1-
  e^{-0.5 \gamma_i}\Big)^L$ where $\gamma_i$ is the
  signal-to-interference-and-noise ratio for player $i$, $R$ is the
  rate at which the wireless system transmits the information and $L$
  is the size of the packets~\cite{SMG99}.}  In such a problem, it is
natural to assume that a device only gets a global information about
the load of the network, and not about each other device which is
connected to the network. This can be expressed using imperfect
monitoring via public signals.

Following~\cite{tomala98} in the framework of repeated matrix games,
we put forward a notion of \emph{public signal} (inspired
by~\cite{tomala98}). A signal will be said public whenever it is
common to all players. That is, after each move, all the players get
the same information (their own action remains of course private). We
will also distinguish several kinds of payoff functions, depending on
whether they are publicly visible (they only depend on the public
signal), or privately visible (they depend on the public signal and on
private actions: the corresponding player knows his payoff!), or
invisible (players may not even be sure of their payoff).

The payoff functions we will focus on in this paper are Boolean
$\omega$-regular payoff functions and mean payoff functions. Some of
the decidability results can be extended in various directions, which
we will mention along the way.

As initial contributions of the paper, we show some undecidability
results, and in particular that the hypothesis of public signal solely
is not sufficient to enjoy all nice decidability results: for mean
payoff functions, which are privately visible, one cannot decide the
constrained existence of a Nash equilibrium. Constrained existence of
a Nash equilibrium asks for the existence of a Nash equilibrium whose
payoff satisfies some given constraint.

The main contribution of the paper is the construction of a so-called
\emph{epistemic game abstraction}. This abstraction is a two-player
turn-based game in which we show that winning strategies of one of the
players (\Eve) actually correspond to Nash equilibria in the original
game. The winning condition for \Eve is rather complex, but can be
simplified in the case of publicly visible payoff functions.  The
epistemic game abstraction is inspired by both the epistemic unfolding
of~\cite{BKP11} used for distributed synthesis, and the suspect game
abstraction of~\cite{BBMU15} used to compute Nash equilibria in
concurrent games with hidden actions.
In our abstraction, we nevertheless not fully formalize epistemic
unfoldings, and concentrate on the structure of the knowledge which is
useful under the assumption of public signals; we show that several
subset constructions (as done initially in~\cite{reif84}, and since
then used in various occasions, see
e.g.~\cite{CDHR07,DEG10,DDG+10,DR11}) made in parallel, are sufficient
to represent the knowledge of all the players. The framework
of~\cite{BBMU15} happens to be a special case of the public signal
monitoring framework of the current paper. This construction can
therefore be seen as an extension of the suspect game abstraction.

This generic construction can be applied to several frameworks with
publicly visible payoff functions. We give two such applications, one
with Boolean $\omega$-regular payoff functions and one with mean
payoff functions.

\paragraph*{Further Related Works.}  We have already discussed
several kinds of related works. Let us give some final remarks on
related works.

We have mentioned earlier that one of the problems for computing Nash
equilibria is to detect deviations and players who deviated. Somehow,
the epistemic game abstraction tracks the potential deviators,
and even though players might not know who exactly is responsible for
the deviation (there might be several suspects), they can try to
punish all potential suspects. And that what we do here.  Very
recently,~\cite{BR17} discusses the detection of deviators, and give
some conditions for them to become common knowledge of the other
players. In our framework, even though deviators may not become fully
common knowledge, we can design mechanisms to punish the relevant
ones.

Recently imperfect information has also been introduced in the setting
of multi-agent temporal logics~\cite{DEG10,DT11,BMM17,BMM+17}, and the
main decidability results assume hierarchical information. However,
while those logics allow to express rich interactions, it can somehow
only consider qualitative properties. Furthermore, no tight complexity
bounds are provided.

In~\cite{brenguier16}, a deviator game abstraction is proposed. It
twists the suspect game abstraction~\cite{BBMU15} to allow for more
general solution concepts (so-called robust equilibria), but it
assumes visibility of actions (hence remove any kind of
uncertainties). Relying on results of~\cite{BR15}, this deviator game
abstraction allows to compute equilibria with mean payoff
functions. Our algorithms for mean payoff functions will also rely on
the polyhedron problem of~\cite{BR15}.

\section{Definitions}

Throughout the paper, if $\bbS \subseteq \bbR$, we write
$\overline{\bbS}$ for $\bbS \cup \{-\infty,+\infty\}$.

\subsection{Concurrent multiplayer games with signals}

We consider the model of concurrent multi-player games, based on the
two-player model of~\cite{AHK02}. This model of games was used for
instance in~\cite{BBMU15}.  We equip games with \emph{signals}, which
will give information to the players.

\begin{definition}
  A \emph{concurrent game with signals} is a tuple
  \[
  \calG = \tuple{V,v_\init,\Agt,\Act,\Sigma,\Allow,\Tab,(\ell_A)_{A \in
      \Agt}, (\payoff_A)_{A \in \Agt}}
  \] 
  where 
  \begin{itemize}
  \item $V$ is a finite set of vertices, 
  \item $v_\init \in V$ is the initial vertex, 
  \item $\Agt$ is a finite set of players, 
  \item $\Act$ is a finite set of actions, 
  \item $\Sigma$ is a finite alphabet, 
  \item $\Allow\colon V \times \Agt \to 2^\Act\setminus\{\emptyset\}$
    is a mapping indicating the actions available to a given player in
    a given vertex, \footnote{This condition ensures that the game is
      non-blocking.}  $\Tab\colon V \times \Act^{\Agt} \to V$
    associates, with a given vertex and a given action tuple the
    target vertex,
  \item for every $A \in \Agt$, $\ell_A \colon \left(\Act^\Agt \times
      V\right) \to \Sigma$ is a signal, and 
  \item for every $A \in \Agt$, $\payoff_A \colon V \cdot
    \left(\Act^\Agt \cdot V\right)^\omega \to \bbD$ is a payoff
    function with values in a domain $\bbD \subseteq \overline{\bbR}$.
  \end{itemize}
  We say that the game has \emph{public signal} if there is $\ell
  \colon \left(\Act^\Agt \times V\right) \to \Sigma$ such that for
  every $A \in \Agt$, $\ell_A = \ell$.
\end{definition}
The signals will help the players monitor the game: for taking
decisions, a player will have the information given by her signal and
the action she played earlier. A public signal will be a common
information given to all the players. Our notion of public signal is
inspired by~\cite{tomala98} and encompasses the model of~\cite{BBMU15}
where only action names were hidden to the players. Note that
monitoring by public signal does not mean that all the players have
the same information: they have private information implied by their
own actions. 

An element of $\Act^\Agt$ is called a move. When an explicit order is
given on the set of players $\Agt = \{A_1,\ldots,A_{|\Agt|}\}$, we
will write a move $m = (m_A)_{A \in \Agt}$ as
$\tuple{m_{A_1},\ldots,m_{A_{|\Agt|}}}$. If $m \in \Act^\Agt$ and $A
\in \Agt$, we write $m(A)$ for the $A$-component of $m$ and $m(-A)$
for all but the $A$ components of $m$. In particular, we write $m(-A)
= m'(-A)$ whenever $m(B) = m'(B)$ for every $B \in \Agt \setminus
\{A\}$. 

A \emph{full history} $h$ in $\calG$ is a finite sequence
\[
v_0 \cdot m_0 \cdot v_1 \cdot m_1 \ldots m_{k-1} \cdot v_k \in V \cdot
\left(\Act^\Agt \cdot V\right)^*
\]
such that for every $0 \le i < k$, $m_i \in \Allow(v_{i})$ and
$v_{i+1} = \Tab(v_i,m_i)$. For readability we will also write $h$ as
\[
v_0 \xrightarrow{m_0} v_1 \xrightarrow{m_1} \ldots
\xrightarrow{m_{k-1}} v_k
\]

We write $\last(h)$ for the last vertex of $h$, that is, $v_k$. If $i
\le k$, we also write $h_{\ge i}$ (resp. $h_{\le i}$) for the suffix
$v_i \cdot m_i \cdot v_{i+1} \cdot m_{i+1} \ldots m_{k-1} \cdot v_k$
(resp. prefix $v_0 \cdot m_0 \cdot v_1 \cdot m_1 \ldots m_{i-1} \cdot
v_i$).  We write $\Hist_\calG(v_0)$ (or simply $\Hist(v_0)$ if $\calG$
is clear in the context) for the set of full histories in $\calG$ that
start at $v_0$.  If $h \in \Hist(v_0)$ and $h' \in \Hist(\last(h))$,
then we write $h \cdot h'$ for the obvious concatenation of both
histories (it then belongs to $\Hist(v_0)$).

Let $A \in \Agt$ be a player. 
The projection of $h$ for $A$ is denoted $\pi_A(h)$ and is defined by
\[
v_0 \cdot (m_1(A),\ell_A(m_1,v_1)) \cdot (m_2(A),\ell_A(m_2,v_2))
\ldots (m_k(A),\ell_A(m_k,v_k)) \in V \times \left(\Act \times
  \Sigma\right)^*
\]
This will be the information available to player $A$: it contains both
the actions she played so far and the signal she received. Note that
we assume perfect recall, that is, while playing, $A$ will remember
all her past knowledge, that is, all of $\pi_A(h)$ if $h$ has been
played so far.  We define the \emph{undistinguishability relation
  $\sim_A$} as the equivalence relation over full histories induced by
$\pi_A$: for two histories $h$ and $h'$, $h \sim_A h'$ iff $\pi_A(h) =
\pi_A(h')$. While playing, if $h \sim_A h'$, $A$ will not be able to
know whether $h$ or $h'$ has been played.  We also define the
$A$-label of $h$ as $\ell_A(h) = \ell_A(m_0,v_1) \cdot \ell_A(m_1,v_2)
\ldots \ell_A(m_{k-1},v_k)$.\footnote{Note that standardly (see
  e.g.~\cite{MW05,BKP11,DR11}), the undistinguishability relation for
  the players is defined according to $\ell_A$, not $\pi_A$ as we do
  here. The difference is that we formally integrate actions of the
  players in the memory. While this actually makes no difference for
  problems like distributed synthesis or even in our framework here
  (where strategies are essentially quantified from the start of the
  game), we believe this is what should be done when studying complex
  interactions between players. For instance this will make a
  difference for subgame-perfect equilibria. We discuss this in
  Appendix~\ref{app:discuss}.}

We extend all the above notions to infinite sequences in a
straightforward way and to the notion of \emph{full play}. We write
$\Plays_\calG(v_0)$ (or simply $\Plays(v_0)$ if $\calG$ is clear in
the context) for the set of full plays in $\calG$ that start at $v_0$.

In the following, we will say that the game $\calG$ has \emph{publicly
  (resp. privately) visible payoffs} if for every $A \in \Agt$, for
every $v_0 \in V$, for every $\rho,\rho' \in \Plays(v_0)$,
$\ell_A(\rho) = \ell_A(\rho')$ (resp. $\rho \sim_A \rho'$) implies
$\payoff_A(\rho) = \payoff_A(\rho')$. Otherwise they are said
\emph{invisible}.  Private visibility of payoffs, while not always
assumed (see for instance~\cite{DDG+10,BMM+17}), are reasonable
assumptions: using only her knowledge, a player knows her
payoff. Public visibility is more restrictive, but will be required
for some of the results.

Let $A \in \Agt$ be a player. A \emph{strategy} for player $A$ from
$v_0$ is a mapping $\sigma_A \colon \Hist(v_0) \to \Act$ such that for
every history $h \in \Hist(v_0)$, $\sigma(h) \in \Allow(\last(h))$. It
is said \emph{$\ell_A$-compatible} whenever furthermore, for all
histories $h,h' \in \Hist(v_0)$, $h \sim_A h'$ implies $\sigma_A(h) =
\sigma_A(h')$. An \emph{outcome} of $\sigma_A$ is a(n infinite) play
$\rho = v_0 \cdot m_0 \cdot v_1 \cdot m_1 \ldots$ such that for every
$i \ge 0$, $\sigma_A(\rho_{\le i}) = m_i(A)$. We write
$\out(\sigma_A,v_0)$ for the set of outcomes of $\sigma_A$ from $v_0$.

A \emph{strategy profile} is a tuple $\sigma_\Agt=(\sigma_A)_{A \in
  \Agt}$, where, for every player $A \in \Agt$, $\sigma_A$ is a
strategy for player $A$. The strategy profile is said
\emph{info-compatible} whenever each $\sigma_A$ is
$\ell_A$-compatible. We write $\out(\sigma_\Agt,v_0)$ for the unique
full play from $v_0$, which is an outcome of all strategies part of
$\sigma_\Agt$.

When $\sigma_\Agt$ is a strategy profile and $\sigma'_A$ a player-$A$
strategy, we write $\sigma_\Agt[A/\sigma'_A]$ for the profile where
$A$ plays according to $\sigma'_A$, and each other player $B$ plays
according to $\sigma_B$. The strategy $\sigma'_A$ is a
\emph{deviation} of player $A$, or an
\emph{$A$-deviation}. 

\begin{definition}
  A \emph{Nash equilibrium} from $v_0$ is an info-compatible strategy
  profile $\sigma$ such that for every $A \in \Agt$, for every
  player-$A$ $\ell_A$-compatible strategy $\sigma'_A$,
  $\payoff_A\Big(\out(\sigma,v_0)\Big) \ge
  \payoff_A\Big(\out(\sigma[A/\sigma'_A],v_0)\Big)$.
\end{definition}
In this definition, deviation $\sigma'_A$ needs not be
$\ell_A$-compatible, since the only meaningful part of $\sigma'_A$ is
along $\out(\sigma[A/\sigma'_A],v_0)$, where there are no
$\sim_A$-equivalent histories: any deviation can be made
$\ell_A$-compatible without affecting the profitability of the
resulting outcome.
Note also that there might be an $A$-deviation $\sigma'_A$ which is
not observable by another player $B$ ($\out(\sigma,v_0) \sim_B
\out(\sigma[A/\sigma'_A],v_0)$), and there might be two deviations
$\sigma'_B$ (by player $B$) and $\sigma'_C$ (by player $C$) that
cannot be distinguished by player $A$ ($\out(\sigma[B/\sigma'_B],v_0)
\sim_A \out(\sigma[C/\sigma'_C],v_0)$).  Tracking such deviations will
be the core of the abstraction we will develop.

\subparagraph*{Payoff functions.}  In the following we will consider
various payoff functions. Let $\Phi$ be an $\omega$-regular property
over some alphabet $\Gamma$. The function $\pay_\Phi \colon
\Gamma^\omega \to \{0,1\}$ is defined by, for every $\mathbf{a} \in
\Gamma^\omega$, $\pay_\Phi(\mathbf{a}) = 1$ if and only if $\mathbf{a}
\models \Phi$.  A publicly (resp. privately) visible payoff function
$\payoff_A$ for player $A$ is said associated with $\Phi$ over
$\Sigma$ (resp. $\Act \times \Sigma$) whenever it is defined by
$\payoff_A(\rho) = \pay_\Phi(\ell_A(\rho))$ (resp. $\payoff_A(\rho) =
\pay_\Phi(\pi_A(\rho)_{-v_0})$, where $\pi_A(\rho)_{-v_0}$ crops the
first $v_0$). Such a payoff function is called a Boolean
$\omega$-regular payoff function.

Let $\Gamma$ be a finite alphabet and $w \colon \Gamma \to \bbZ$ be a
weight assigning a value to every letter of that alphabet. We define
two payoff functions over $\Gamma^\omega$ by, for every $\mathbf{a} =
(a_i)_{i \ge 1} \in \Gamma^\omega$, $\pay_{\underline{\MP}_w}
(\mathbf{a}) = \liminf_{n \to \infty} \frac{1}{n} \sum_{i=1}^n w(a_i)$
and $\pay_{\overline{\MP}_w} (\mathbf{a}) = \limsup_{n \to \infty}
\frac{1}{n} \sum_{i=1}^n w(a_i)$.  A publicly visible payoff function
$\payoff_A$ for player $A$ is said associated with the liminf
(resp. limsup) mean payoff of $w$ whenener it is defined by
$\payoff_A(\rho) = \pay_{\underline{\MP}_w}(\ell_A(\rho))$
(resp. $\pay_{\overline{\MP}_w} (\ell_A(\rho))$). A privately visible
payoff function $\payoff_A$ for player $A$ is said associated with the
liminf (resp. limsup) mean payoff of $w$ whenener it is defined by
$\payoff_A(\rho) = \pay_{\underline{\MP}_w}(\pi_A(\rho)_{-v_0})$
(resp. $\pay_{\overline{\MP}_w}(\pi_A(\rho)_{-v_0})$).
 
A payoff function $\payoff \colon V \times (\Act^\Agt \times V)^\omega
\to \bbD$ is said \emph{prefix-independent} whenever for every full
play $\rho$, for every suffix $\rho_{\ge i}$ of $\rho$, $\payoff(\rho)
= \payoff(\rho_{\ge i})$.

\input{figure}

\begin{example}
  We now illustrate most notions on the game of Fig.~\ref{fig:ex}.
  This is a game with three players $A_1$, $A_2$ and $A_3$, and which
  is played basically in two steps, starting at $v_0$.  Graphically an
  edge labelled $\tuple{a_1,a_2,a_3}$ between two vertices $v$ and
  $v'$ represents the fact that $a_i \in \Allow(v,A_i)$ for every $i
  \in \{1,2,3\}$ and that $v'=\Tab(v,\tuple{a_1,a_2,a_3})$.  As a
  convention, $*$ stands for both $a$ and $b$.  For readability,
  bottom vertices explicitely indicate the payoffs of the three
  players (same order as for actions) if the game ends in that vertex.

  After the first step of the game, signal yellow or green is sent to
  all the players.
  Histories $v_0 \cdot \tuple{a,b,a} \cdot v_2$ and $v_0 \cdot
  \tuple{a,a,a} \cdot v_1$ are
  undistinguishable 
  by $A_1$ and $A_3$ (same action, same signal), but they can be
  distinguished by $A_2$ because of different actions (even if the
  same signal is observed).

  In bold red, we have depicted a strategy profile, which is actually
  a Nash equilibrium. We analyze the possible deviations in this game
  to argue for this.
  \begin{itemize}
  \item First there is an $A_2$-deviation to $v_1$. This deviation is
    invisible to both players $A_1$ and $A_3$. For this reason, the
    strategy out of $v_1$ for $A_1$ is to play $a$ (same as out of
    $v_2$).
    On the other hand, even though this would be profitable to her,
    $A_1$ cannot deviate from $v_1$, since we are in a branch where
    $A_2$ has already deviated, and at most one player is allowed to
    deviate at a time (and anyway $A_1$ does not know that they are in
    state $v_1$).
  \item There is an $A_1$-deviation from $v_2$ to $0,1,0$, which is
    not profitable to $A_1$.
  \item On the other hand, there is no possible deviation to $v_3$,
    since this would require two players to change their actions
    simultaneously ($A_1$ and $A_2$).
  \item Then, there is an $A_1$-deviation to $v_4$ and another
    $A_3$-deviation to $v_5$; both activate the green signal. $A_2$
    knows there has been a deviation (because of the green signal),
    but she doesn't know who has deviated and whether the game
    proceeds to $v_4$ or $v_5$ (but she knows that if $A_1$ has
    deviated, then we are in $v_4$, and if $A_3$ has deviated, we are
    in $v_5$). Then, $A_2$ has to find a way to punish both players,
    to be safe.  On the other hand, both players $A_1$ and $A_3$
    precisely know what has happened: in case she didn't deviate
    herself, she knows the other one deviated! And she knows in which
    state the game is in. Hence in state $v_4$, $A_3$ can help player
    $A_2$ punishing $A_1$, whereas in state $v_5$, $A_1$ can help
    player $A_2$ punishing $A_3$. Examples of punishing moves are
    therefore those depicted in red and bold; and they are part of the
    global strategy profile.
    Note that the action of $A_2$ out of $v_5$ has to be the same as
    the one out of $v_4$: this is required given the imperfect
    knowledge of $A_2$. On the other hand, the action of $A_3$ can be
    different out of $v_4$ and out of $v_5$ (which is the case in the
    given example profile).
  \end{itemize}
\end{example}

\begin{remark}
  To be fully formal, we use rather heavy notations for concurrent
  games. We will use them for the main technical construction of the
  paper (the epistemic game abstraction), but we will allow ourselves
  simpler notations in other proofs (undecidability and complexity
  reductions). For instance, instead of fully defining components
  $\Allow$ and $\Tab$ for a concurrent game, we will better write
  transitions $q \xrightarrow{\tuple{a_1,a_2,a_3}} q'$ to denote the
  fact that $a_i \in \Allow(q,A_i)$ ($i \in \{1,2,3\}$) and
  $q'=\Tab(q,\tuple{a_1,a_2,a_3})$.
\end{remark}

\subsection{Two-player turn-based game structures}

They are specific cases of the previous model, where at each vertex,
at most one player has more than one action in her set of allowed
actions. But for convenience, we will give a simplified definition,
with only objects that will be useful.
A two-player turn-based game structure is a tuple $G =
\tuple{S,S_\Eve,S_\Adam,s_\init,A,\Allow,\Tab}$, where $S = S_\Eve
\sqcup S_\Adam$ is a finite set of states (states in $S_\Eve$ belong
to player \Eve whereas states in $S_\Adam$ belong to player \Adam),
$s_\init \in S$ is the initial state, $A$ is a finite alphabet,
$\Allow \colon S \to 2^A \setminus \{\emptyset\}$ gives the set of
available actions, and $\Tab \colon S \times A \to S$ is the
next-state function. If $s \in S_\Eve$ (resp. $S_\Adam$), $\Allow(s)$
is the set of actions allowed to \Eve (resp. \Adam) in state $s$.

In this context, strategies will see sequences of states and actions,
with full information. Note that we do not include any winning
condition or payoff function in the tuple, hence the name structure.

\subsection{The problems we are looking at}

We are interested in the constrained existence of a Nash
equilibrium. For simplicity, we define constraints using non-strict
thresholds constraints, but could well impose more involved
constraints, like Boolean combinations of linear constraints.

\begin{problem}[Constrained existence problem]
  Given a game with signals $\calG = \langle
  V,v_\init,\Agt,\Act,\Sigma,$ $\Allow,\Tab,(\ell_A)_{A \in \Agt},
  (\payoff_A)_{A \in \Agt} \rangle$ and threshold vectors $(\nu_A)_{A
    \in \Agt}$, $(\nu'_A)_{A \in \Agt} \in \overline{\bbQ}^\Agt$, can
  we decide whether there exists a Nash equilibrium $\sigma_\Agt$ from
  $v_\init$ such that for every $A \in \Agt$, $\nu_A \le
  \payoff_A(\out(\sigma_\Agt,v_\init)) \le \nu'_A$?  If so, compute
  one.
  If the constraints on the payoff are trivial (that is, $\nu_A =
  -\infty$ and $\nu'_A = +\infty$ for every $A \in \Agt$), we simply
  speak of the existence problem.
\end{problem}



\subsection{First undecidability results}

\subparagraph{Undecidability of the problem for general signals and
  publicly visible Boolean $\omega$-regular objectives.}

Using the model of concurrent games with signals we can express the
grand-coalition problem in games with imperfect information
of~\cite{BK10,BKP11} (which is strongly related to the distributed
synthesis problem~\cite{PR90}). We can therefore easily infer
undecidability of the constrained existence and then of the existence
problem, even for simple Boolean payoff functions and few players.

\begin{theorem}
  The existence problem in games with signals is undecidable with
  three players and publicly visible Boolean $\omega$-regular payoff
  functions.
\end{theorem}

\begin{proof}
  A game of imperfect information as defined in~\cite{BK10} can be
  seen as a concurrent game, with an extra player to resolve the
  non-determinism allowed in the former model, and signals that only
  depend on the current visited vertex. Following this modelization,
  the grand-coalition problem then asks whether there are
  info-compatible strategies for the original players (they form the
  so-called ``grand-coalition''), such that for every strategy of the
  extra player, the outcome of the resulting strategy profile
  satisfies the given winning condition (for instance an
  $\omega$-regular winning condition).  It is known that the
  grand-coalition question is undecidable, already for reachability
  properties and two players~\cite{BK10}. Such a reachability can be
  obviously made publicly visible to the players (by revealing it when
  it is reached).

  We will explain how this can be coded into a constrained existence
  problem over the same game, with three players and Boolean
  $\omega$-regular payoff functions.
  We assign to each of the original players a Boolean payoff
  function corresponding to their original $\omega$-regular winning
  condition $\Phi$, and to the extra player a Boolean payoff
  function corresponding to the negation of $\Phi$. The
  grand-coalition has a winning strategy in the initial game if and
  only if there is a Nash equilibrium in the new game, where the
  original players have a payoffs $1$ (which implies that the extra
  player has payoff $0$, but cannot improve). This shows the
  undecidability of the constrained existence problem.

  \medskip Using a trick already used in~\cite{BBMU15}, we can extend
  this undecidability proof to the existence problem as follows. Let
  $\calG$ be the game discussed above with objective $\Phi$, and build
  $\calG'$ as on the next picture. Original players appear first in
  the action tuples, say they are ordered as $\Agt =
  \{A_1,\ldots,A_N\}$, and the extra player is at the end, called
  $A_{\text{extra}}$.
  
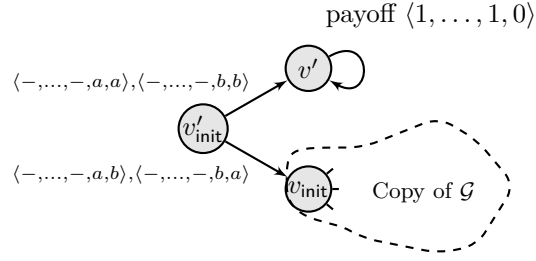
\begin{center}
    \begin{tikzpicture}[thick]
      \tikzset{noeud/.style={circle,draw=black,thick,fill=black!10,minimum
          height=6mm,inner sep=0pt}}
      \draw (0,0) node [noeud] (A) {$v'_\init$};
      \draw (30:1.6cm) node [noeud] (B) {$v'$};
      \draw (3,1.5) node  {payoff $\tuple{1,\ldots,1,0}$};
      \draw (-30:1.6cm) node [noeud] (C) {$v_\init$};
      \begin{scope}[xshift=1.1cm,yshift=-.75cm]
        \draw[dashed, rounded corners=2mm] (0,0) .. controls +(90:5mm) .. (1,.5) -- (2,.9) -- (2.5,.6) --
        (3,0) -- (2.5,-.8) -- (2,-.9) -- (1,-.7) .. controls +(150:5mm) and
        +(-90:5mm) .. (0,0);
        \draw (1.8,-0.1) node {\begin{minipage}{2.1cm}\footnotesize\centering
            Copy of~$\calG$
          \end{minipage}};
      \end{scope}
      \draw[-latex'] (A) -- (B) node[midway,above left=-2pt] {$\scriptstyle \tuple{-,\ldots,-,a,a}, \tuple{-,\ldots,-,b,b}$};
      \draw[-latex'] (A) -- (C) node[midway,below left=-2pt] {$\scriptstyle \tuple{-,\ldots,-,a,b}, \tuple{-,\ldots,-,b,a}$};
      \draw[-latex'] (B) .. controls +(30:10mm) and +(-30:10mm) .. (B);
      \draw[dashed] (C) -- +(40:5mm);
      \draw[dashed] (C) -- +(0:4mm);
      \draw[dashed] (C) -- +(-40:5mm);
    \end{tikzpicture}
  \end{center}
  The new initial vertex is $v'_\init$, and from there, $A_1$ and
  $A_{\text{extra}}$ play a matching-penny game (symbol ``$-$''
  indicates that the other players have a single action, with no
  impact). The new payoff functions are as follows:
  \begin{itemize}
  \item Players $A_i$ ($1 \le i \le N$) get $1$ whenever the game ends
    in $v'$, or gets what $\payoff_{A_i}$ was giving $A_i$ if the game
    proceeds to $\calG$;
  \item Player $A_{\text{extra}}$ gets $0$ whenever the game ends in
    $v'$, or gets what $\payoff_{A_{\text{extra}}}$ was giving her if
    the game proceeds to $\calG$.
  \end{itemize}

  We will check that there is a Nash equilibrium in $\calG'$ if and
  only if there is a Nash equilibrium where all players except
  $A_{\text{extra}}$ gets $1$ in $\calG$.
  Indeed, assume that $\sigma_\Agt$ is a Nash equilibrium in $\calG$
  where all players except $A_{\text{extra}}$ gets $1$. Then playing
  $\tuple{-,\ldots,-,a,b}$ followed by $\sigma_\Agt$ is a Nash
  equilibrium in $\calG'$.

  Conversely assume there is a Nash equilibrium $\sigma_\Agt$ in
  $\calG'$. Assume that the outcome of $\sigma_\Agt$ goes to $\calG$.
  Since $A_N$ has no profitable deviation, this means that $A_N$ has
  payoff $1$ along that outcome (otherwise $A_N$ would deviate to
  $v'$). Then the strategy profile after $v_\init$ is a suitable Nash
  equilibrium in $\calG$.  Assume now that the outcome of
  $\sigma_\Agt$ goes to $v'$. Since $A_{\text{extra}}$ has no
  profitable deviation, this means in particular that the strategy
  profile which plays according to $\sigma_\Agt$ after $v_\init$ is a
  suitable Nash equilibrium in $\calG$.
\end{proof}

\subparagraph{Undecidability of the problem for public signals and
  private visible mean payoff functions.}

We show here that public payoff functions (hence somehow public
signals) are necessary to have our main result
(Theorem~\ref{theo:MP}).  We argue this in the framework of
mean payoff functions.

\begin{theorem}
  \label{theo:MP-undec3}
  The constrained existence problem in games with a public signal is
  undecidable with two players and privately visible mean payoff
  functions.
\end{theorem}

\begin{proof}
  To prove this, we use the undecidability of blind mean-payoff
  games~\cite{DDG+10}. In a two-player turn-based game where the first
  player is blind (that is, it can only observe that actions happen,
  but not get any information on the states the game is visiting), we
  cannot decide whether the first player can achieve a positive limsup
  (or liminf) mean payoff. Note: in this game, the first player cannot
  see the encountered weights. Somehow, the payoff function is
  invisible.

  Let $G$ be a mean-payoff game, in which \Eve is assumed to be
  blind. We assume that the payoff function of \Eve is a limsup payoff
  function given by weight $w$ (that is, $\overline\MP_w$). We assume
  ``$-$'' is a fresh symbol not used in game $G$. The alphabet of
  actions of $G'$ is the one fo $G$, plus that fresh symbol. Further
  assume that $W$ is the maximal absolute weight value appearing in
  the game.
  We construct the new concurrent game $G'$ as follows:
  \begin{itemize}
  \item it has two players, $A_1$ and $A_2$
  \item if $v$ belongs to \Eve in $G$, then $v \xrightarrow{a} v'$ is
    transformed into an edge $v \xrightarrow{\tuple{a,a}} v'$ in $G'$,
    and edges $v \xrightarrow{\tuple{a,b}} \textsf{lost}$ if $ a\ne
    b$. We define $\ell(\tuple{a,a},v') = -$ and
    $\ell(\tuple{a,b},\textsf{lost}) = \text{lost}$
  \item if $v$ belongs to \Adam in $G$, then $v \xrightarrow{a} v'$ is
    transformed into an edge $v \xrightarrow{\tuple{-,a}} v'$ in $G'$;
    we set $\ell(v \xrightarrow{\tuple{-,a}} v') = -$
  \item we have an additional edge $\textsf{lost}
    \xrightarrow{\tuple{-,-}} \textsf{lost}$, and
    $\ell(\tuple{-,-},\textsf{lost}) = \text{lost}$
  \item we write $\textit{transform}(e)$ for the set of transformed
    edges from $e$ as described above
  \item $w_{A_1}$ assigns $0$ to every edge
  \item for every edge $e$ in $G$, for every $e' \in
    \textit{transform}(e)$, $w_{A_2}(\textit{transform}(e)) = -w(e)$
  \item for $e' = \textsf{lost} \xrightarrow{\tuple{-,-}}
    \textsf{lost}$, we set $w_{A_2}(e') = -W-1$
  \end{itemize}
  We assume that the payoff for player $A_2$ is given by
  $\underline{\MP}_{w_{A_2}}$ (liminf). The choice for $A_1$ does not
  matter (the payoff of all plays is $0$)...  Note that the signal is
  public and that the payoff functions are privately visible for the
  respective players.  Also note that $A_1$ has no incentive to
  deviate (since the weight is $0$ everywhere), and that player $A_2$
  has no better strategy than to copy $A_1$ from (former) \Eve-states
  (otherwise the game proceeds to state $\textsf{lost}$ and $A_2$ gets
  the worst he can get). The only way to improve his payoff is
  therefore to deviate in some previous \Adam-state.

  \smallskip We show that \Eve has a winning strategy in the original
  blind game $G$ to achieve $\overline{\MP}_w>0$ if and only if there
  is a Nash equilibrium in the new game $G'$ where
  $\underline{\MP}_{w_{A_2}}<0$.
  Let $\sigma_\Agt = (\sigma_{A_1},\sigma_{A_2})$ be a Nash
  equilibrium in $G'$ achieving $\underline{\MP}_{w_{A_2}}<0$. Let
  $\rho = \out(\sigma_\Agt,v_0)$ be the main outcome of $\sigma_\Agt$
  ($v_0$ is the initial vertex). Then $\overline{\MP}_w(\rho) =
  -\underline{\MP}_{w_{A_2}}(\rho)>0$. Furthermore, for every
  $\sigma'_{A_2}$, if $\rho' =
  \out(\sigma_\Agt[A_2/\sigma'_{A_2}],v_0)$, then
  $\underline{\MP}_{w_{A_2}}(\rho') \le
  \underline{\MP}_{w_{A_2}}(\rho) <0$ and
  $\overline{\MP}_w(\rho')>0$. In particular, if \Eve follows the
  strategy of $A_1$ in the original $G$, she will ensure
  $\overline{\MP}_w>0$.
  
  Conversely pick a winning strategy $\sigma_\Eve$ for \Eve ensuring
  $\overline{\MP}_w>0$. Since \Eve is blind, $\sigma_\Eve$ can be seen
  as a single word.  Against that (fixed) strategy, \Adam has an
  optimal counter-strategy $\sigma_\Adam$ in $G$ (the possible choices
  of \Adam, given the strategy of \Eve, can be represented as an
  infinite tree, and at any stage he can choose the best subtree;
  hence \Adam has an optimal strategy!). At each length of the joint
  outcome prefix, a single choice has to be made, hence this
  \Adam-strategy can be made blind (we define it as a single word,
  representing the choives along the joint outcome); we extend the
  strategy to other histories by enforcing $\ell$-compatibility.
  Assume $A_1$ plays according to $\sigma_\Eve$ and $A_2$ according to
  $\sigma_\Adam$.  Let $\rho$ be the outcome of these two strategies.
  Then $A_2$ has no profitable deviation (since \Adam was playing
  optimally against $\sigma_\Eve$) and since $\overline{\MP}_w(\rho)
  >0$, it is the case that $\underline{\MP}_{w_{A_2}}(\rho)<0$. So,
  this is a Nash equilibrium satisfying the required constraint.
\end{proof}

\subparagraph*{Discussion.}  While the first undecidability result is
not very surprising, since allowing arbitrary private signals can
really complexify the structure of players' knowledge, we believe that
the second result justifies a restriction to public payoff functions
to get decidability results. We also believe that the second result is
an argument for restricting to public signal. Indeed, a hierarchical
signal, as standardly done in distributed synthesis, will only make
sense with privately visible payoff functions, hence this will be
undecidable.

In the following we will focus on public signals and develop an
epistemic game abstraction, which will record and track possible
deviations in the game. This will then be applied to two frameworks
with publicly visible payoff functions.

\section{The epistemic game abstraction}

Building over~\cite{BBMU15} and~\cite{BKP11}, we construct an
epistemic game, which will record possible behaviours of the system,
together with possible unilateral deviations. In~\cite{BKP11}, notions
of epistemic Kripke structures are used to really track the precise
knowledge of the players. These are mostly useful since
undistinguishable states (expressed using signals here) are assumed
arbitrary (no hierarchical structure). We could do the same here, but
we think that would be overly complex and hide the real structure of
knowledge in the framework of public signals. We therefore prefer to
stick to simpler subset constructions, which are more commonly used
(see e.g.~\cite{reif84} or later~\cite{CDHR07,DDG+10,DR11}), though it
has to be a bit more involved here since also deviations have to be
tracked.

Let $\calG =
\tuple{V,v_\init,\Agt,\Act,\Sigma,\Allow,\Tab,\ell,(\payoff_A)_{A\in\Agt}}$
be a concurrent game with public signal.  We will first define the
epistemic abstraction as a two-player game structure $\calE_\calG =
\tuple{S_\Eve,S_\Adam, s_\init,\Sigma',\Allow', \Tab'}$, and then
state the correspondence between $\calG$ and $\calE_\calG$. The
epistemic abstraction will later be used for decidability and
algorithmics purposes.

For clarity, we use the terminology ``vertices'' in $\calG$ and
``states'' (or ``epistemic states'') in $\calE_\calG$.

\subsection{Construction of the Game Structure $\calE_\calG$}

The game $\calE_\calG$ will be played between two players, \Eve and
\Adam. The aim of \Eve is to build a suitable Nash equilibrium,
whereas the aim of \Adam is to prove that it is not an equilibrium; in
particular, \Adam will try to find a profitable deviation (to disprove
the claim of \Eve that she is building a Nash equilibrium). Choices
available to \Eve and \Adam in the abstract game have to reflect
partial knowledge of the players in the original game $\calG$. States
in the abstract game will therefore store information, which will be
sufficient to infer the undistinguishability relation of all the
players in the original game. Thanks to the public signal assumption,
this information will be simple enough to have a simple structure.

In the following, we set $\Agt^\bot = \Agt \cup \{\bot\}$, where
$\bot$ is a fresh symbol.  For convenience, if $m \in \Act^\Agt$, we
extend the notation $m(-A)$ when $A \in \Agt$ to $\Agt^\bot$ by
setting $m(-\bot) = m$. We now describe all the components of
$\calE_\calG$.

A state of \Eve will store a set of vertices of the original game one
can be in, together with possible deviators. More precisely, states of
\Eve are defined as $S_\Eve = \{s \colon \Agt^\bot \to 2^V \mid
|s(\bot)| \le 1\}$. Let $s \in S_\Eve$. If $A \in \Agt$, vertices of
$s(A)$ are those where the game can be in, assuming one has followed
the suggestions of \Eve so far, up to an $A$-deviation; on the other
hand, if $s(\bot) \ne \emptyset$, the single vertex $v \in s(\bot)$ is
the one the game is in, assuming one has followed all suggestions by
\Eve so far (in particular, if \Eve is building a Nash equilibrium,
then this vertex belongs to the main outcome of the equilibrium).  We
define $\sit(s) = \{(v,A) \in V \times \Agt^\bot\mid v \in s(A)\}$ for
the set of \emph{situations} the game can be in at $s$:
\begin{enumerate}[(a)]
\item $(v,\bot) \in \sit(s)$ is the situation where the game has
  proceeded to vertex $v$ without any deviation;
\item $(v,A) \in \sit(s)$ with $A \in \Agt$ is the situation where the
  game has proceeded to vertex $v$ benefitting, from an $A$-deviation.
\end{enumerate}
Structure of state $s$ will allow to infer the undistinguishability
relation of all the players in game $\calG$: basically (and we will
formalize this later), if she is not responsible for a deviation,
player $A \in \Agt$ will not know in which of the situations of
$\sit(s) \setminus V \times \{A\}$ the game has proceeded; if she is
responsible for a deviation, player $A$ will know exactly in which
vertex $v \in s(A)$ the game has proceeded.

Let $s \in S_\Eve$. From state $s$, \Eve will suggest a tuple of moves
$M$, one for each possible situation $(v,A) \in \sit(s)$. This tuple
of moves has to satisfy the undistinguishability relation: if a player
does not distinguish between two situations, her action should be the
same in these two situations:
\begin{multline*}
  \Allow'(s) = \Big\{M \in \prod_{(v,A) \in \sit(s)} \Allow(v) \mid
  \forall (v_B,B), (v_C,C) \in \sit(s),\\[-.3cm] \forall A \in \Agt
  \setminus \{B,C\},\ M(v_B,B)(A) = M(v_C,C)(A)\Big\}
\end{multline*}
In the above set, the constraint $M(v_B,B)(A) = M(v_C,C)(A)$ expresses
the fact that player $A$ should play the same action in the two
situations $(v_B,B)$ and $(v_C,C)$, since she does not distinguish
between them. Obviously, we assume $\Sigma'$ contains all elements of
$\Allow'(s)$ above.

States of \Adam are then copies of states of \Eve with suggestions
given by \Eve, that is: $S_\Adam = \{(s,M) \mid s \in S_\Eve \times
\Allow'(s)\}$. And naturally, we define $\Tab'(s,M) = (s,M)$ if $M \in
\Allow'(s)$.

Let $(s,M) \in S_\Adam$. From state $(s,M)$, \Adam will choose a
signal value which can be activated from some situation allowed in
$s$, after no deviation or a single-player deviation w.r.t. $M$.
From a situation $(v,A) \in \sit(s)$ with $A \in \Agt$, only
$A$-deviations can be allowed (since we look for unilateral
deviations), hence any signal activated by an $A$-deviation
(w.r.t. $M(v,A)$) from $v$ should be allowed. From the situation
$(v,\bot) \in \sit(s)$ (if there is one), one can continue without any
deviation, or any kind of single-player deviation should be allowed,
hence the signal activated by $M(v,\bot)$ from $v$ should be allowed,
and any signal activated by some $A$-deviation (w.r.t. $M(v,\bot)$)
from $v$ should be allowed as well.  Formally:
\begin{multline*}
  \Allow'(s,M) = \left\{\beta \in \Sigma \begin{array}{c} \mid
      \\[-.2cm] \mid \\[-.2cm] \mid \\[-.2cm] \mid \\[-.2cm] \mid
      \\[-.2cm] \end{array} \begin{array}{l} \exists A \in \Agt \\
      \exists v \in s(A) \\
      \exists m \in \Act^\Agt\\
    \end{array}\ \text{s.t.}\
    \begin{array}{l@{~}l@{}} 
      \text{(i)} & m(-A) = M(v,A)(-A) \\
      \text{(ii)} & \ell(m,\Tab(v,m)) = \beta \end{array}\right\} \\
  \cup \left\{\beta \in \Sigma \begin{array}{c} \mid
      \\[-.2cm] \mid \\[-.2cm] \mid \\[-.2cm] \mid \\[-.2cm] \mid
      \\[-.2cm] \end{array} \begin{array}{l} \exists v \in s(\bot) \\
      \exists m \in \Act^\Agt\\
      \exists A \in \Agt \end{array}\ \text{s.t.}\
    \begin{array}{l@{~}l@{}} 
      \text{(i)} & m(-A) = M(v,\bot)(-A) \\ 
      \text{(ii)} & \ell(m,\Tab(v,m)) = \beta \end{array}\right\} 
\end{multline*}
Note that we implicitly assume that $\Sigma'$ contains $\Sigma$.

It remains to explain how one can compute the next state of some
$(s,M) \in S_\Adam$ after some signal value $\beta \in
\Allow'(s,M)$. The new state has to represent the new knowledge of the
players in the original game when they have seen signal $\beta$; this
has to take into account all possible deviations that we have already
discussed which activate the signal value $\beta$. The new state is
the result of several simultaneous subset constructions, which we
formalize as follows:
$s' = \Tab'((s,M),\beta)$, where for every $A \in \Agt^\bot$, $v' \in
s'(A)$ if and only if there is $m \in \Act^\Agt$ such that $\beta =
\ell(m,v')$, and 
\begin{enumerate}
\item either there is $v \in s(A)$ such that $m(-A) = M(v,A)(-A)$ and $v'
  = \Tab(v,m)$;
\item or there is $v \in s(\bot)$ such that $m(-A) = M(v,\bot)(-A)$
  and $v' = \Tab(v,m)$.
\end{enumerate}
Note that in case $A = \bot$, the two above cases are redundant.

\medskip Before stating properties of $\calE_\calG$, we illustrate the
construction.

\input{figure-epistemic}

\begin{example}
We consider again the example of Fig.~\ref{fig:ex}, and we assume that
the public signal when reaching the leaves of the game is uniformly
orange. We depict (part of) the epistemic game abstraction of the game
on Fig.~\ref{fig:epistemic}. 
One can notice that from \Eve-states $s_1$ and $s_2$, moves are
multi-dimensional, in the sense that there is one move per vertex
appearing in the state. There are nevertheless compatibility
conditions which should be satisfied (expressed in condition
$\Allow'$); for instance, from $s_2$, player $A_2$ does not
distinguish between the two options (i) $A_1$ has deviated and the
game is in $v_4$, and (ii) $A_3$ has deviated and the game is in
$v_5$, hence the action of player $A_2$ should be the same in the two
moves ($a$ in the depicted example, written in red).
\end{example}

\subsection{Interpretation of this abstraction}

We gave an intuitive meaning to the (epistemic) states of
$\calE_\calG$, we now need to formalize it. To that aim, we need to
explain how full histories and plays in $\calE_\calG$ can be
interpreted as full histories and plays in $\calG$.
 
Let $v_0 \in V$, and define $s_0 \colon \Agt^\bot \to 2^V \in S_\Eve$
such that $s_0(\bot) = \{v_0\}$ and $s_0(A) = \emptyset$ for every $A
\in \Agt$. In the following, when $M \in \Allow'(s)$ for some $s \in
S_\Eve$, if we speak of some $M(v,A)$, we implicitly assume that
$(v,A) \in \sit(s)$.  Given a full history $H = s_0 \xrightarrow{M_0}
(s_0,M_0) \xrightarrow{\beta_0} s_1 \xrightarrow{M_1} (s_1,M_1)
\xrightarrow{\beta_1} s_2 \ldots (s_{k-1},M_{k-1})
\xrightarrow{\beta_{k-1}} s_k$ in $\calE_\calG$, we write
$\concrete(H)$ for the set of full histories in the original game,
which correspond to $H$, up to a single deviation, that is: $v_0
\xrightarrow{m_0} v_1 \xrightarrow{m_1} v_2 \ldots v_{k-1}
\xrightarrow{m_{k-1}} v_k \in \concrete(H)$ whenever for every $0 \le
i \le k-1$, $v_{i+1} = \Tab(v_i,m_i)$ and $\beta_i =
\ell(m_i,v_{i+1})$, and:
\begin{enumerate}[(a)]
\item either $m_i = M_i(v_i,\bot)$ for every $0 \le i \le k-1$;
\item or there exist $A \in \Agt$ and $0 \le i_0 \le k-1$ such that
  \begin{enumerate}[(i)]
  \item for every $0 \le i < i_0$, $m_i = M_i(v_i,\bot)$;
  \item $m_{i_0} \ne M_{i_0}(v_{i_0},\bot)$, but $m_{i_0}(-A) =
    M_{i_0}(v_{i_0},\bot)(-A)$;
  \item for every $i_0 < i \le k-1$, $m_i(-A) = M_i(v_i,A)(-A)$.
  \end{enumerate}
\end{enumerate}
Case (a) corresponds to a concrete history with no deviation (all
moves suggested by \Eve have been followed). Case (b) corresponds to a
deviation by player $A$, and $i_0$ is the position at which player $A$
has started deviating.

We write $\concrete_\bot(H)$ for the set of histories of type (a);
there is at most one such history, which is the real concrete history
suggested by \Eve. And we write $\concrete_A(H)$ for the set of
histories of the type (b) with deviator $A$.  
%
We extend all these notions to full plays. A full play visiting only
\Eve-states $s$ such that $s(\bot) \ne \emptyset$ is called a
$\bot$-play. 

We will now state several properties, which will explicit the
information stored in the (states of the) epistemic game abstraction,
and make clear how one can infer the undistinguishability relations of
the various players thanks to that information.

The two following lemmas can be easily infered from the definitions:

\begin{lemma}
  The set $\concrete_\bot(H)$ contains at most one element, and it has
  exactly one element if and only if $s_k(\bot) \ne \emptyset$. If
  $\concrete_\bot(H) = \{v_0 \xrightarrow{m_0} v_1 \xrightarrow{m_1}
  v_2 \ldots v_{k-1} \xrightarrow{m_{k-1}} v_k\}$, it is the case that
  $s_i(\bot) = \{v_i\}$ for every $0 \le i \le k$.
\end{lemma}

\begin{lemma}
  Let $A \in \Agt$. Then:
  \begin{itemize}
  \item Assume $v_0 \xrightarrow{m_0} v_1 \xrightarrow{m_1} v_2 \ldots
    v_{k-1} \xrightarrow{m_{k-1}} v_k \in \concrete_A(H)$ with index
    $i_0$. Then for every $0 \le i \le i_0$, $s_i(\bot) = \{v_i\}$,
    and for every $i_0<i \le k-1$, $v_i \in s_i(A)$.
  \item For every $v_k \in s_k(A)$, there is some $h \in
    \concrete_A(H)$ such that $\last(h) = v_k$.
  \end{itemize}
\end{lemma}

The set $\concrete_\bot(H)$, if non-empty, stores the real history of
the game, if no deviation has occurred. On the other hand, there can
be deviations of one of the players; they will appear in
$\concrete_A(H)$ (meaning that $A$ has deviated). As long as states of
$H$ have a non-empty $\bot$-component, those deviations will somehow
be invisible to other players, whereas when we will reach a state with
no $\bot$-component, those deviations will partly be visible to other
players.

We now explain how the structure of $\calE_\calG$ explicit the
undistinguishability relations of the players.

\begin{lemma}
  \label{lemma:1}\label{lemma:2}
  Let $A \in \Agt$, and pick $h_1 \ne h_2 \in \concrete(H)$.  Then
  $h_1 \sim_A h_2$ if and only if $h_1,h_2 \notin \concrete_A(H)$.
\end{lemma}

\begin{proof}
  By setting $i_0=k+1$, one can treat the unique history in
  $\concrete_\bot(H)$ similarly to other histories in $\concrete_B(H)$
  for $B \in \Agt$. 

  For $j \in \{1,2\}$, we let $B_j \in \Agt^\bot \setminus \{A\}$ such
  that $h_j \in \concrete_{B_j}(H)$ with index $i_0^j$. We write $h_j
  = v_0 \xrightarrow{m^j_0} v_1^j \xrightarrow{m^j_1} v_2^j \ldots
  \xrightarrow{m^j_{k-1}} v_k^j$. By construction, we have (i) for
  every $i<i^j_0$, $m^j_i = M_i(v_i^j,\bot)$, (ii) $m^j_{i_0^j}(-B_j)
  = M_{i_0^j}(v_{i_0^j}^j,\bot)(-B_j)$, (iii) for every $i>i_0^j$,
  $m_i^j(-B_j) = M_i(v_i^j,B_j)(-B_j)$.  Also, for every $i$,
  $\ell(m^j_i,v_{i+1}^j) = \beta_i$.

  Player $A$ only sees $h_j$ through its $\pi_A$-projection, that is:
  \[
  \begin{array}{r} \pi_A(h_j) \\ \phantom{v}
  \end{array} \begin{array}{@{}l} =~ v_0 \cdot
    (m^j_0(A),\ell(m^j_0,v_1^j)) \cdot (m^j_1(A),\ell(m^j_1,v_2^j))
    \ldots
    (m^j_{k-1}(A),\ell(m^j_{k-1},v_k^j)) \\
    =~ v_0 \cdot (m^j_0(A),\beta_0) \cdot (m^j_1(A),\beta_1) \ldots
    (m^j_{k-1}(A),\beta_{k-1})
    \end{array}
  \]
  Note that, writing $i_0$ for $\min(i_0^1,i_0^2)$, we have:
  \begin{enumerate}[(a)]
  \item for every $0 \le i < i_0$, $m^1_i = m^2_i$, hence $(h_1)_{\le
      i_0} = (h_2)_{\le i_0}$;
  \item $h_1 \sim_A h_2$ if and only if for every $0 \le i<k$,
    $m^1_i(A) = m^2_i(A)$.
  \end{enumerate}
  W.l.o.g. we assume $i_0^1 \le i_0^2$.
  \begin{itemize}
  \item First assume that $B_1 = A$ (i.e. $h_1 \in \concrete_{A}(H)$).
    It is the case that $m^1_{i_0^1}(A) \ne
    M_{i_0^1}(v_{i_0^1}^1,\bot)(A)$.  Towards a contradiction, we
    assume $h_1 \sim_A h_2$. It implies $m^2_{i_0^2}(A) =
    m^1_{i_0^1}(A) \ne M_{i_0^1}(v_{i_0^1}^1,\bot)(A)$. Hence, $B_2=A$
    as well, and $i_0^1 = i_0$. We deduce that $(h_1)_{\le i_0+1} =
    (h_2)_{\le i_0+1}$. By induction, it is now easy to infer from
    (iii) above and from the fact that $m_i^1(A) = m_i^2(A)$ for every
    $i$ (since $h_1 \sim_A h_2$), that $h_1 = h_2$, which yields a
    contradiction. Thus, it cannot be the case, under hypothesis
    $B_1=A$, that $h_1 \sim_A h_2$.
  \item Next assume that $B_1 \ne A$, $i_0^1<i_0^2$ and $B_2 = A$
    (i.e. $h_2 \in \concrete_{A}(H)$). By definition of $\Allow'$
    (since $A \ne B_1$), it is the case that
    $M_{i_0^2}(v_{i_0^2}^1,B_1)(A) =
    M_{i_0^2}(v_{i_0^2}^2,\bot)(A)$. Now it is the case that
    $m_{i_0^2}^1(A) = M_{i_0^2}(v_{i_0^2}^1,B_1)(A)$ while
    $m_{i_0^2}^2(A) \ne M_{i_0^2}(v_{i_0^2}^2,\bot)(A)$. Hence, $h_1
    \sim_A h_2$. 
  \item Then assume that $B_1 \ne A$ and $B_2 \ne A$. By definition of
    $\Allow'$, we have that $M_i(v_i^1,B_1)(A) = M_i(v_i^2,B_2)(A)$
    for every $i$. In particular, since $m_i^1(A) = M_i(v_i^1,B_1)(A)$
    and $m_i^2(A) = M_i(v_i^2,B_2)(A)$ for every $i$, we deduce that
    $h_1 \sim_A h_2$.
  \end{itemize}
  We have therefore shown the expected equivalence.
\end{proof}

Assuming public visibility of the original payoff functions in
$\calG$, we can define when $R$ is a full play in $\calE_\calG$, and
$A \in \Agt$, $\payoff'_A(R) = \payoff_A(\rho)$, where $\rho \in
\concrete(R)$.

\begin{lemma}
  \label{lemma:payoff}
  Assuming public visibility of the original payoff functions in
  $\calG$, the payoff function $\payoff'_A$ is well-defined.
\end{lemma}

\begin{proof}
  Write $R = s_0 \xrightarrow{M_0} (s_0,M_0) \xrightarrow{\beta_0} s_1
  \ldots s_k \xrightarrow{M_k} (s_k,M_k) \xrightarrow{\beta_k} s_{k+1}
  \ldots$. It is sufficient to notice that $\rho,\rho' \in
  \concrete(R)$ imply $\ell(\rho) = \ell(\rho')$. Hence for every $A
  \in \Agt$, $\payoff_A(\rho) = \payoff_A(\rho')$. Therefore,
  $\payoff'_A(R)$ does not depend on the choice of the witness $\rho$.
\end{proof}

\subsection{Winning condition of \Eve}
A zero-sum game will be played on the game structure $\calE_\calG$,
and the winning condition of \Eve will be given on the branching
structure of the set of outcomes of a strategy for \Eve, and not
individually on each outcome, as standardly in two-player zero-sum
games.  We write $s_\init$ for the state of \Eve such that
$s_\init(\bot) = \{v_\init\}$ and $s_\init(A)=\emptyset$ for every $A
\in \Agt$. Let $p = (p_A)_{A \in \Agt} \in {\overline{\bbR}}^\Agt$,
and $\sigma_\Eve$ be a strategy for \Eve in $\calE_\calG$; it is said
\emph{winning} for $p$ from $s_\init$ whenever $\payoff(\rho) = p$,
where $\rho$ is the unique element of
$\concrete_\bot(\out_\bot(\sigma_\Eve,s_\init))$ (where we write
$\out_\bot(\sigma_{\Eve},s_\init)$ for the unique outcome of
$\sigma_{\Eve}$ from $s_\init$ which is a $\bot$-play), and for every
$R \in \out(\sigma_\Eve,s_\init)$, for every $A \in \Agt$, for every
$\rho \in \concrete_A(R)$, $\payoff_A(\rho) \le p_A$.

For every epistemic state $s \in S_\Eve$, we define the set of
\emph{suspect} players $\suspect(s) = \{A \in \Agt \mid s(A) \ne
\emptyset\}$ (this is the set of players that may have deviated). By
extension, if $R = s_0 \xrightarrow{M_0} (s_0,M_0)
\xrightarrow{\beta_0} s_1 \ldots s_k \xrightarrow{M_k} (s_k,M_k)
\xrightarrow{\beta_k} s_{k+1} \ldots$, we define $\suspect(R) =
\lim_{k \to \infty} \suspect(s_k)$. Note that the sequence
$(\suspect(s_k))_{k}$ is non-increasing, hence it stabilizes.

Assuming public visibility of the payoff functions in $\calG$, we can
define when $R$ is a full play in $\calE_\calG$, and $A \in \Agt$,
$\payoff'_A(R) = \payoff_A(\rho)$, where $\rho \in \concrete(R)$. It
is easy to show that $\payoff'_A$ is well-defined for every $A \in
\Agt$.
Under this assumption, the winning condition of \Eve can be rewritten
as: $\sigma_\Eve$ is winning for $p$ from $s_\init$ whenever
$\payoff'(\out_\bot(\sigma_\Eve,s_\init)) = p$, and for every $R \in
\out(\sigma_\Eve,s_\init)$, for every $A \in \suspect(R)$,
$\payoff'_A(R) \le p_A$.

\subsection{Correction of the epistemic abstraction}

The epistemic abstraction tracks everything that is required to detect
Nash equilibria in the original game, which we make explicit in the
next result. Note that this theorem does not require public visibility
of the payoff functions.

\begin{theorem}
  \label{theo:correction}
  Let $\calG$ be a concurrent game with public signal, and $p \in
  \overline{\bbR}^\Agt$.  There is a Nash equilibrium in $\calG$ with
  payoff $p$ from $v_\init$ if and only if \Eve has a winning strategy
  for $p$ in $\calE_\calG$ from $s_\init$.
\end{theorem}

The proof of this theorem highlights a correspondence between Nash
equilibria in $\calG$ and winning strategies of \Eve in $\calE_\calG$.
In this correspondence, the main outcome of the equilibrium in $\calG$
is the unique $\bot$-concretisation of the unique $\bot$-play
generated by the winning strategy of \Eve.
 
\begin{proof}
  We fix an initial vertex $v_0$ in $\calG$, and we let $s_0$ be its
  corresponding initial state in $\calE_\calG$ (that is, $s_0(\bot) =
  \{v_0\}$ and for every $A \in \Agt$, $s_0(A) =\emptyset$.

  Assume first that $\sigma_\Agt = (\sigma_A)_{A \in \Agt}$ is a Nash
  equilibrium in the original game $\calG$ from $v_0$ with payoff
  $p$. We define a strategy $\sigma_\Eve$ for \Eve in the epistemic
  game $\calE_\calG$, which we will show is winning from $s_0$ for
  $p$.

  Let $H = s_0 \xrightarrow{M_0} (s_0,M_0) \xrightarrow{\beta_0} s_1
  \xrightarrow{M_1} (s_1,M_1) \xrightarrow{\beta_1} s_2 \ldots
  \xrightarrow{M_{k-1}} (s_{k-1},M_{k-1}) \xrightarrow{\beta_{k-1}}
  s_k$ be a full history in $\calE_\calG$. Fix $A \in \Agt$. For every
  $B \in \Agt^\bot \setminus \{A\}$, for every $h \in \concrete_B(H)
  \setminus \concrete_A(H)$, we set $\sigma_\Eve(H)(\last(h),B)(A) =
  \sigma_A(h)$. If some coordinate is not definable that way, we set
  it arbitrarily. 

  We need to show that this is well-defined. Pick $h,h' \in
  \concrete(H) \setminus \concrete_A(H)$. By Lemma~\ref{lemma:1}, $h
  \sim_A h'$.  Hence we deduce that $\sigma_A(h) = \sigma_A(h)$, which
  implies that $\sigma_\Eve$ is well-defined and that $\sigma_\Eve(H)
  \in \Allow'(\last(H))$. 

  We now show that $\sigma_\Eve$ is a winning strategy from $s_0$ for
  $p$.  First notice that $\out(\sigma_\Agt,v_0)$ coincides with the
  unique element $\rho_\bot$ of
  $\concrete_\bot(\out_\bot(\sigma_\Eve,s_0))$, hence
  $\payoff(\rho_\bot) = p$.  Let $R = s_0 \xrightarrow{M_0} (s_0,M_0)
  \xrightarrow{\beta_0} s_1 \xrightarrow{M_1} (s_1,M_1)
  \xrightarrow{\beta_1} s_2 \ldots$ be an (infinite) outcome of
  $\sigma_\Eve$.  Pick $\rho \in \concrete(R)$ such that $\rho \ne
  \rho_\bot$. There is $A \in \Agt$ such that $\rho \in
  \concrete_A(R)$, with index $k$. In particular we can write $\rho =
  v_0 \xrightarrow{m_0} v_1 \dots \xrightarrow{m_{k-1}} v_k\ldots$
  such that (i) for every $i <k$, $m_i = M_i(v_i,\bot)$, (ii) $m_k(-A)
  = M_k(v_k,\bot)(-A)$ and $m_k(A) \ne M_k(v_k,\bot)(A)$, (iii) for
  every $i>k$, $m_i(-A) = M_i(v_i,A)(-A)$; and $\ell(m_i,v_{i+1}) =
  \beta_i$ for every $i$. From $\rho$ we design a deviation
  $\sigma'_A$ for player $A$ in $\calG$ as follows:
  \begin{itemize}
  \item for every prefix $\rho_{\le i}$ of $\rho$,
    $\sigma'_A(\rho_{\le i}) = m_i(A)$;
  \item we extend $\sigma'_A$ to histories that are undistinguishable
    by $A$ from some $\rho_{\le i}$;
  \item we set $\sigma'_A(h) = \sigma_A(h)$ for all other histories
    $h$.
  \end{itemize}
  By construction, this is a well-defined info-compatible strategy for
  player $A$ in $\calG$. We use it as a deviation
  w.r.t. $\sigma_\Agt$.  Since $\sigma_\Agt$ is a Nash equilibrium, it
  is the case that $\payoff_A(\out(\sigma_\Agt[A/\sigma'_A],v_0)) \le
  \payoff_A(\out(\sigma_\Agt,v_0)) = p_A$.
  By construction we have $\rho = \out(\sigma_\Agt[A/\sigma'_A],v_0)$,
  and therefore $\payoff_A(\rho) \le p_A$.  We conclude that
  $\sigma_\Eve$ is a winning strategy.

  \medskip Conversely consider a winning strategy $\sigma_\Eve$ from
  $s_0$ in the epistemic game $\calE_\calG$. When we build the
  strategy profile $\sigma_\Agt$, we need not only to check the main
  outcome, but also to somehow foresee the possible deviations.
  We build $\sigma_\Agt$ inductively on the length of histories. We
  maintain the following invariant: For all histories of length at
  most $i$, which can be obtained as generated by $\sigma_\Agt$ or by
  a deviation by a unique player, we know a (partial) function $f_i
  \colon \Hist_{\calG} \to \Hist_{\calE_\calG}$ such that:
  \begin{enumerate}[(i)]
  \item $f_i$ extends $f_{i-1}$ (if $i>0$), and $\dom(f_i)$ only
    contains histories of length at most $i$;
  \item if $h$ is a prefix of $\out(\sigma_\Agt,v_0)$ or a prefix of a
    single-player deviation $\out(\sigma_\Agt[A/\sigma'_A],v_0)$ of
    length at most $i$, then $h \in \dom(f_i)$;
  \item for every $h \in \dom(f_i)$, $h \in \concrete(f_i(h))$ and
    $f_i(h)$ is a prefix of an outcome of $\sigma_\Eve$;
  \item if $h$ is a prefix of $\out(\sigma_\Agt,v_0)$, then $f_i(h)$
    is a prefix of $\out_\bot(\sigma_\Eve,s_0)$ and $h$ is the unique
    element of $\concrete_\bot(f_i(h))$;
  \item if $h$ is a prefix of $\out(\sigma_\Agt[A/\sigma'_A],v_0)$ for
    some player-$A$ deviation $\sigma'_A$, then $f_i(h)$ is a prefix
    of some outcome in $\out(\sigma_\Eve,s_0)$ and $h \in
    \concrete_A(f_i(h))$;
  \item for every $A \in \Agt$, for every $h \sim_A h'$, if $h
    \in\dom(f_i)$, then $h' \in \dom(f_i)$ and $f_i(h) = f_i(h')$.
  \end{enumerate}

  The base case $i=0$ is obvious, since there is a single history of
  length $0$ both in $\calG$ and in $\calE_\calG$. 
  We assume we have proven the result for $i$. We first initialize
  $f_{i+1}$ as $f_i$ on $\dom(f_i)$. Then, we consider a length-$i$
  outcome $h$ of either $\sigma_\Agt$ (no deviation so far) or of
  $\sigma_\Agt[B/\sigma'_B]$ for some (real)\footnote{In this context,
    a real deviation is one which does generate an outcome different
    from the main outcome without deviation.} $B$-deviation
  $\sigma'_B$, from $v_0$.  For every $A \in \Agt$, we define
  $\sigma_A(h) = \sigma_\Eve(f_i(h))(\last(h),\bot)(A)$ if $h$ is an
  outcome of $\sigma_\Agt$, and $\sigma_A(h) =
  \sigma_\Eve(f_i(h))(\last(h),B)(A)$ is the result of a (real)
  $B$-deviation w.r.t $\sigma_\Agt$. If there is only one such $h$,
  then this is well-defined.

  Otherwise, pick two such distinct outcomes $h$ and $h'$ such that $h
  \sim_A h'$. By (vi), $f_i(h) = f_i(h')$, which we denote $H$. By
  Lemma~\ref{lemma:2}, $h,h' \in \concrete(H) \setminus
  \concrete_A(H)$. Let $B,B' \in \Agt^\bot \setminus \{A\}$ such that
  $h \in\concrete_B(H)$ and $h' \in \concrete_{B'}(H)$. Now, since
  $\sigma_\Eve(H) \in \Allow'(\last(H))$, it is the case that
  $\sigma_\Eve(H)(\last(h),B)(A) =
  \sigma_\Eve(H)(\last(h'),B')(A)$. This implies that $\sigma_A(h) =
  \sigma_A(h')$, hence the profile $\sigma_\Agt$ is properly defined
  (values of $\sigma_\Agt$ outside the above domain is not important
  and can be set arbitrarily).

  Further define $f_{i+1}$ as follows:
  \begin{itemize}
  \item Let $h$ be the length-$i$ outcome of $\sigma_\Agt$ from
    $v_0$. Then the history $h' = h \xrightarrow{\sigma_\Agt(h)} v$ is
    the length-$(i+1)$ outcome of $\sigma_\Agt$ (where we have written
    $v$ for $\Tab(\last(h),\sigma_\Agt(h))$): we set $f_{i+1}(h') =
    \left(f_i(h) \xrightarrow{\sigma_\Eve(f_i(h))}
      (s,\sigma_\Eve(f_i(h)))\xrightarrow{\beta} s' \right)$, where $s
    = \last(f_i(h))$, $\beta = \ell(\sigma_\Agt(h),v)$ and $s' =
    \Tab((s,\sigma_\Eve(f_i(h))),\beta)$.
  \item Let $h$ be a length-$i$ outcome of $\sigma_\Agt$ or of some
    real $A$-deviation $\sigma_\Agt[A/\sigma'_A]$ from $v_0$. Then for
    every $m' \in \Act^\Agt$ such that $m'(-A) = \sigma_\Agt(h)(-A)$,
    $h' = h \xrightarrow{m'} v$ is a length-$(i+1)$ outcome of some
    $A$-deviation $\sigma_\Agt[A/\sigma'_A]$ (where we have written
    $v$ for $\Tab(\last(h),\sigma_\Agt(h))$): we set $f_{i+1}(h') =
    \left(f_i(h) \xrightarrow{\sigma_\Eve(f_i(h))}
      (s,\sigma_\Eve(f_i(h)))\xrightarrow{\beta} s' \right)$, where $s
    = \last(f_i(h))$, $\beta = \ell(m',v)$ and $s' =
    \Tab((s,\sigma_\Eve(f_i(h))),\beta)$.
  \end{itemize}

  Conditions (i)--(v) are now easy to get, by definition of the
  various objects. It remains to show condition (vi).
  Assume that $h_j \xrightarrow{m_j} v_j$ (with $j=1,2$) are two
  different length-$(i+1)$ outcomes in the domain of $f_{i+1}$, and
  assume that $\left(h_1 \xrightarrow{m_1} v_1\right) \sim_A \left(h_2
    \xrightarrow{m_2} v_2\right)$ for some $A \in \Agt$.  Then:
  \begin{itemize}
  \item $h_1 \sim_A h_2$, hence $f_i(h_1) \stackrel{\text{\upshape
        (vi)}}{=} f_i(h_2)$;
  \item $m_1(A) = m_2(A)$ and $\ell(m_1,v_1) = \ell(m_2,v_2)$.
  \end{itemize}
  By construction, 
  \[
  f_{i+1}(h_j) = \left( f_{i}(h_j) \xrightarrow{\sigma_\Eve(f_i(h_j))}
    (s_j,\sigma_\Eve(f_i(h_j))) \xrightarrow{\beta_j} s'_j \right)
  \]
  where $s_j = \last(f_{i}(h_j))$, $\beta_j = \ell(m_j,v_j)$ and $s'_j
  = \Tab((s_j,\sigma_\Eve(f_i(h_j))),\beta_j)$. This value is the same
  for both $j=1$ and $j=2$. Hence point (vi) is now proven.
  This concludes the induction step. 

  We now show that $\sigma_\Agt$ is a Nash equilibrium from $v_0$.
  We first define the function $f$ as the limit of the sequence
  $(f_i)_{i \in \bbN}$: $f$ can speak of finite histories and infinite
  plays.  Let $\rho = \out(\sigma_\Agt,v_0)$. By (iv) characterizing
  $f$, we get that $f(\rho) = \out_\bot(\sigma_\Eve,s_0)$ and $\rho$
  is the unique element of $\concrete_\bot(f(\rho))$. Since
  $\sigma_\Eve$ is winning with payoff $p$, we get that $p =
  \payoff(\rho)$.

  Consider some $A$-deviation $\sigma'_A$, and write $\rho' =
  \out(\sigma_\Agt[A/\sigma'_A],v_0)$. By (ii) and (v), we deduce
  $\rho'\in \dom(f)$, $f(\rho') \in \out(\sigma_\Eve,s_0)$ and $\rho'
  \in \concrete_A(f(\rho'))$.  Since $\sigma_\Eve$ is a winning
  strategy for \Eve, we get that $\payoff_A(\rho') \le p_A$. That is,
  there is no profitable deviation for every $A \in \Agt$:
  $\sigma_\Agt$ is a Nash equilibrium with payoff $p$. 
\end{proof}

\subsection{Remarks on the construction}

We did not formalize the epistemic unfolding as it is made
in~\cite{BKP11}. We believe we do not really learn anything for public
signal using it. And the above extended subset construction can much
better be understood.

One could argue that this epistemic game gives more information to the
players, since \Eve explicitely gives to everyone the move that should
be played. But in the real game, the players also have that
information, which is obtained by an initial coordination of the
players (this is required to achieve equilibria).

Finally, notice that the espitemic game constructed here generalizes
the suspect game construction of~\cite{BBMU15}, where all players have
perfect information on the states of the game, but cannot see the
actions that are precisely played. Somehow, games in~\cite{BBMU15}
have a public signal telling the state the game is in (that is,
$\ell(m,v) =v$). So, in the suspect game of~\cite{BBMU15}, the sole
uncertainty is in the players that may have deviated, not in the set
of states that are visited.
 
\begin{remark}
  Let us analyze the size of the epistemic game abstraction.
  \begin{itemize}
  \item The size of the alphabet is bounded by $|\Sigma|+
    |\Act|^{|\Agt| \cdot |V| \cdot (1+|\Agt|)}$. Since $|\Sigma|$ is
    bounded by $|V| \cdot |\Act|^{|\Agt|}$, the size of the alphabet
    is in $O(|\Act|^{|\Agt|^2 \cdot |V|})$
  \item The number of states is therefore in $O(2^{|\Agt|\cdot
      |V|}\cdot |\Act|^{|\Agt|^2 \cdot |V|})$. 
  \end{itemize}
  The epistemic game is therefore of exponential size w.r.t. the
  initial game. Note that we could reduce the bounds by using tricks
  like those in~\cite[Prop.~4.8]{BBMU15}, but this would not avoid an
  exponential blowup.
\end{remark}

\section{Algorithmics for publicly visible payoffs}

While the construction of the epistemic game has transformed the
computation of Nash equilibria in a concurrent game with public signal
to the computation of winning strategies in a two-player zero-sum
turn-based game, we cannot apply standard algorithms out-of-the-box,
because the winning condition is rather complex. In the following, we
present two applications of that approach in the context of publicly
visible payoffs, one with Boolean payoff functions, and another with
mean payoff functions. Remember that in the latter case, public
visibility is required to have decidability
(Theorem~\ref{theo:MP-undec3}).
 
We let $\calG = \tuple{V,v_\init,\Agt,\Act,\Sigma,\Allow,\Tab,\ell,
  (\payoff_A)_{A \in \Agt}}$ be a concurrent game with public signal,
and we let $\calE_\calG = \tuple{S_\Eve,S_\Adam, s_\init,\Act^\Agt
  \cup \Sigma,\Allow', \Tab'}$ be its epistemic game abstraction.  We
write $S_\bot = \{s \in S_\Eve \mid s(\bot) \ne \emptyset\}$. The
$\bot$-part of $\calE_\calG$ is the set of states $S_\bot \cup \{(s,M)
\in S_\Adam \mid s \in S_\bot\}$.

In this section, otherwise specified, we assume public visibility and
prefix-independence of payoff functions.  We recall this notion here:
A payoff function $\payoff \colon V \times (\Act^\Agt \times V)^\omega
\to \bbD$ is said \emph{prefix-independent} whenever for every full
play $\rho$, for every suffix $\rho_{\ge i}$ of $\rho$, $\payoff(\rho)
= \payoff(\rho_{\ge i})$.

Thanks to Lemma~\ref{lemma:payoff} (see page~\pageref{lemma:payoff}),
one can therefore attach to every full play $R$ of $\calE_\calG$, a
$|\Agt|$-dimensional payoff vector $\payoff'(R) = (\payoff'_A(R))_{A
  \in \Agt}$.  We then recall here the winning condition for \Eve in
$\calE_\calG$: a strategy $\sigma_\Eve$ is winning whenever there
exists some $p = (p_A)_{A \in \Agt}$ such that
$\payoff'(\out_\bot(\sigma_\Eve,s_0)) = p$, and for every $R \in
\out(\sigma_\Eve,s_0)$, for every $A \in \suspect(R)$, $\payoff'_A(R)
\le p_A$.

\subsection{First reduction}

We define the value of epistemic state $s$ as follows:
\[
\Value(s) = \{w \in \overline{\bbR}^\Agt \mid \exists \sigma_\Eve
\forall \rho \in \out(\sigma_\Eve,s)\ \forall A \in \suspect(\rho),\
\payoff'_A(\rho) \le w(A)\}
\]
The set $\Value(s)$ is upward-closed, and we want to compute a
representation thereof.  Note that for all players $A \in \Agt
\setminus \suspect(s)$, one can represent possible values for $A$ by
$-\infty$.

Somehow, the set $\Value(s)$ stores the possible payoffs than one
should achieve along the main outcome of a candidate Nash equilibrium
when a deviation to $s$ may happen. Otherwise, one of the suspect
players will manage to deviate and have a larger payoff. This is
formalized as follows.

\begin{lemma}
  \label{lemma:botplay}
  \Eve has a winning strategy in $\calE_\calG$ from $s_0$ for payoff
  vector $p$ if and only there is a $\bot$-play $R = s_0
  \xrightarrow{M_0} (s_0,M_0) \xrightarrow{\beta_0} s_1 \ldots s_{k-1}
  \xrightarrow{M_{k-1}} (s_{k-1},M_{k-1}) \xrightarrow{\beta_{k-1}}
  s_k \ldots$ from $s_0$, such that $p = (\payoff'_A(R))_{A \in
    \Agt}$, and for every history $s_0 \xrightarrow{M_0} (s_0,M_0)
  \xrightarrow{\beta_0} s_1 \ldots s_{i-1} \xrightarrow{M_{i-1}}
  (s_{i-1},M_{i-1}) \xrightarrow{\beta'_{i-1}} s$, $p \in \Value(s)$.
\end{lemma}

In the above lemma, $s_0 \xrightarrow{M_0} (s_0,M_0)
\xrightarrow{\beta_0} s_1 \ldots s_{i-1} \xrightarrow{M_{i-1}}
(s_{i-1},M_{i-1}) \xrightarrow{\beta'_{i-1}} s$ corresponds to a
prefix of $R$ with a one-step deviation from an \Adam-state.

\begin{proof}
  Assume $\sigma_\Eve$ is a winning strategy for \Eve in
  $\calE_\calG$. Let $R = \out_\bot(\sigma_\Eve,s_0)$ be the main
  outcome of $\sigma_\Eve$ from $s_0$, and $p$ its payoff. For all $R'
  \in \out(\sigma_{\Eve},s_0)$, for every $A \in \suspect(R)$,
  $\payoff'_A(R) \le p_A$.

  Pick some history $H = s_0 \xrightarrow{M_0} (s_0,M_0)
  \xrightarrow{\beta_0} s_1 \ldots s_{i-1} \xrightarrow{M_{i-1}}
  (s_{i-1},M_{i-1}) \xrightarrow{\beta'_{i-1}} s$ generated by
  $\sigma_\Eve$, where all but the last move follow $R$.  Consider the
  strategy $\sigma'_\Eve$, such that for every history $H'$ starting
  at $s$, $\sigma'_\Eve(H') = \sigma_\Eve(H \cdot H')$. Then, for
  every $R' \in \out(\sigma'_\Eve,s)$, $H \cdot R' \in
  \out(\sigma_\Eve,s_0)$. Hence, since $\sigma_\Eve$ is winning, for
  every $R' \in \out(\sigma'_\Eve,s)$, for every $A \in \suspect(R')$,
  $\payoff'_A(H \cdot R') \le p_A$. Since $\payoff'_A$ is
  prefix-independent, $\payoff'_A(R') = \payoff'_A(H \cdot R') \le
  p_A$, hence $p \in \Value(s)$.

  \medskip Conversely assume $R$ is a $\bot$-play with payoff $p$, and
  such that it satisfies the hypotheses of the lemma, and for every
  state $s$ resulting from a single-move deviation of $R$, let
  $\sigma^s$ be a strategy witnessing the fact that $p \in \Value(s)$.
  Define $\sigma_\Eve$ as a strategy following $R$ along the prefixes
  of $R$, and which plays according to $\sigma^s$ when \Adam deviates
  from $R$ at state $s$. Then it is not difficult to see that this is
  a winning strategy for \Eve in $\calE_\calG$. 
\end{proof}

Together with Theorem~\ref{theo:correction}, the above lemma reduces
the problem of computing Nash equilibria (for prefix-independent
payoff functions) in the original game to:
\begin{enumerate}[(i)]
\item the computation of (representatives of) the sets $\Value(s)$
  when $s \in S_\Eve \setminus S_\bot$
\item the search of a $\bot$-play satisfying the constraints of the
  lemma.
\end{enumerate}

\subsection{Generic procedure to compute the value sets for
  prefix-independent payoffs}
\label{subsec:generic}

The game $\calE_\calG$ has a specific structure, in which the set of
suspects along a play is non-increasing (and stabilizing...), as soon
as we leave the (main) $\bot$-part of the game. We will use this
structure to design a generic bottom-up procedure to compute the
values of states in the non-$\bot$-part of $\calE_\calG$.

We make the following hypothesis. We have an oracle which, given a
tuple $G = \tuple{\tuple{S,S_\Eve,S_\Adam,s_\init,A,\Allow,\Tab},
  S_\o, \gamma, \delta}$, where the first component is a turn-based
game structure, $S_\o \subseteq S$ is a set of output states such that
for every $s \in S_\o$, for every $a \in \Allow(s)$, $\Tab(s,a)
=s$,\footnote{We could instead say that $s$ is blocking, in that there
  is no allowed action, but this would not fit our formal definition
  of a game. Hence we use this trick here.}  $\delta \colon S_\o \to
\Up(\overline{\bbR}^d)$ assigns an upward-closed set of
$\overline{\bbR}^d$ to every output state, and $\gamma \colon \Plays_G
\to \overline{\bbR}^d$ is a $d$-dimensional payoff function), computes
at any state of the game, the set
\[
\val^G(s) = \left\{w \in \overline{\bbR}^d \begin{array}{c} \mid
    \\[-.2cm] \mid \\[-.2cm] \mid \end{array} \exists \sigma_\Eve
  \forall R \in \out(\sigma_\Eve,s)\ \begin{array}{l} R\ \text{not
      visiting $S_\o$ implies}\ \gamma(R) \le
    w,\ \text{and}\\
    R\ \text{ending up in $s' \in S_\o$ implies}\ w \in
    \delta(s') \end{array}\right\}
\]

We will define subgames of $\calE_\calG$, where sets of suspects will
be fixed. The idea is that, within such a game, the payoff reduces to
a multi-dimensional standard payoff function, with nodes connected to
lower components (where the set of suspects is reduced). Then using a
bottom-up computation, we will compute the value sets in all states of
$S_\Eve \setminus S_\bot$ in $\calE_\calG$.

We formalize this idea now.  Inductively, we define for every $\calS
\subseteq \Agt$ a core game $\Core_\calS =
\tuple{\tuple{S^\calS,S_\Eve^\calS,S_\Adam^\calS,s_\init^\calS,\Act^\Agt
    \cup \Sigma,\Allow'_{|S_\Eve^\calS \cup
      S_\Adam^\calS},\Tab^\calS_{|S_\Eve^\calS \cup
      S_\Adam^\calS}},S_\o^\calS,\gamma^\calS,\delta^\calS}$ as
follows.
\begin{itemize}
\item $S^\calS = S_\Eve^\calS \cup S_\Adam^\calS$
\item $S_\Eve^\calS = \{s \in S_\Eve\setminus S_\bot \mid \suspect(s)
  = S\} \cup \{s' \in S_\Eve \mid \calS \ne \suspect(s) \subseteq
  \calS\ \text{and}\ \exists s \in S_\Adam\ \text{with}\
  \suspect(s)=\calS\ \text{and}\ s' = \Tab'(s,m)\ \text{for some}\ m\
  \text{in}\ \calE_\calG\}$
\item $S_\Adam^\calS = \{s \in S_\Adam \setminus S_\bot\mid
  \suspect(s) = S\} \cup \{s' \in S_\Adam \mid \calS \ne \suspect(s)
  \subseteq \calS\ \text{and}\ \exists s \in S_\Eve\ \text{with}\
  \suspect(s)=\calS\ \text{and}\ s' = \Tab'(s,m)\ \text{for some}\ m\
  \text{in}\ \calE_\calG\}$
\item $S_\o^\calS = \{s \in S^\calS \mid \suspect(s) \ne \calS\}$
\item for every infinite play $R$ in $\Core_\calS$ not visiting
  $S_\o^\calS$, for every $A \in \calS$, we set $\gamma^\calS(R) =
  \payoff'_A(R)$ (this is a $|\calS|$-dimensional payoff function)
\item for every history ending in $s \in S_\o^\calS$, 
  \[
  \delta^\calS(s) = \{ w \in \overline{\bbR}^{\calS} \mid
  w_{|\suspect(s)} \in \val^{\Core_{\suspect(s)}}(s)\}
  \]
\end{itemize}

Correctness of the approach is stated below:

\begin{lemma}
  Let $s \in S \setminus S_\bot$. Then $s$ is a state of the core game
  $\Core_{\suspect(s)}$. And 
  \[
  \Value(s) = \{w \in \overline{\bbR}^\Agt \mid w_{|\suspect(s)} \in
  \val^{\Core_{\suspect(s)}}(s)\}
  \]
\end{lemma}

\begin{proof}
  The proof is by induction on the set $\calS$ of suspects.

  The base case is a state $s$ of $\Core_{\calS}$ such that $\calS$ is
  minimal.  For that $\calS$, $S^\calS_\o = \emptyset$.  Seing $s$ as
  a state of $\calE_\calG$, the set $\Value(s)$ can be seen as the
  value set of the multiple-dimension payoff function $(\payoff'_A)_{A
    \in \calS}$, which is exactly what the oracle computes for us.

  We assume that for every set $\calS \ne \calS' \subseteq \calS$, we
  have shown that for every $s'$ such that $\suspect(s') = \calS'$,
  $\Value(s') = \{w \in \overline{\bbR}^\Agt \mid w_{|\calS'} \in
  \val^{\Core_{\calS'}}(s')\}$.

  Pick a strategy $\sigma_\Eve$ in $\calE_\calG$ from $s$ that
  witnesses some $w \in \Value(s)$.  Pick an outcome $R$ of
  $\sigma_\Eve$ which leaves $\Core_{\suspect(s)}$ (viewed as a
  subarena of $\calE_\calG$) at state $s' \in S_\o^{\suspect(s)}$
  after prefix $H$. Write $\calS' = \suspect(s')$. Then $\sigma_\Agt$
  after prefix $H$ achieves $w$ as well (by prefix-independence): $w
  \in \Value(s')$ By induction hypothesis, $w_{|\calS'} \in
  \val^{\Core_{\calS'}}(s')$, hence $w_{|\calS} \in \delta^\calS(s')$.
  Hence $\sigma_\Eve$ (pruned when leaving $\Core_{\calS}$) achieves
  $\val^{\Core_{\calS}}(s)$ (since all outcomes which do not
  leave the current core are no problem for the winning condition). We
  conclude that 
  \[
  \Value(s) \subseteq \{w \in \overline{\bbR}^\Agt \mid w_{|\suspect(s)} \in
  \val^{\Core_{\suspect(s)}}(s)\}
  \]

  \medskip Conversely, pick a strategy $\sigma_\Eve$ in
  $\Core_{\suspect(s)}$ from $s$ which witnesses $w_{|\calS} \in
  \val^{\Core_{\suspect(s)}}(s)$.  Pick an outcome $R$ of
  $\sigma_\Eve$ which leaves at $s' \in S^\calS_\o$, and let $h$ be
  the smallest prefix of $R$ reaching $s'$. Write $\calS' =
  \suspect(s')$. Then, $w_{|\calS} \in \delta^\calS(s')$, that is,
  $w_{|\calS'} \in \val^{\Core_{\calS'}}(s')$. In particular, by
  induction hypothesis, we get that $w \in \Value(s')$. For every
  state $s' \in S^\calS_\o$ reachable via $\sigma_\Eve$, we let
  $\sigma^{s'}$ be a strategy in $\calE_\calG$ which witnesses $w$
  from $s'$. Then, let $\sigma'_\Eve$ be the strategy that plays
  according to $\sigma_\Eve$ as long as we do not visit $S^\calS_\o$,
  and switches to $\sigma^{s'}$ as soon as we visit $s' \in
  S^\calS_\o$. This is not difficult to show that this strategy
  achieves $w$! 
\end{proof}

\subsection{Boolean $\omega$-regular objectives}

For prefix-independent $\omega$-regular objectives, we will use the
generic approach presented previously.  We assume that each payoff
function $\payoff_A$ is associated with a prefix-independent Boolean
($\omega$-regular) objective $\Omega_A$ over $\Sigma$. The payoff
value domain $\bbD$ is $\{0,1\}$. Hence, for readability we can assume
that various values are taken in $\{0,1\}$ as well (instead of
$\overline{\bbR}$).

We consider a core game $\Core_{\calS}$ and a state $s$ such that
$\suspect(s) = \calS$.  We have that $(1)_\calS \in
\val^{\Core_{\calS}}(s)$.\footnote{This is a notation for the vector
  assigning $1$ to every $A \in \calS$.}
The set of possible improvements of those values is finite (try to
assign $0$ instead of $1$ to every player $A \in \calS$). 

For each $\calS' \subseteq \calS$, we write $w[\calS']$ for the vector
$(0)_{A \in \calS'} \cup (1)_{A \in \calS \setminus \calS'}$.  Then:
$w[\calS'] \in \val_{\Core_{\calS}}(s) = \Value(s)$ if and only if
\Eve has a strategy to enforce
\begin{equation}
  \label{win}
  \Big(\bigcap_{A \in \calS'}
  \Omega_A^c\Big) \cup \Reach(\{s' \in S_\o^{\calS} \mid w[\calS'] \in
  \delta^\calS(s')\}) 
\end{equation}
from $s$, where $\Reach(\ldots)$ indicates a reachability
objective. That is, \Eve should play such that every play either
leaves via an output node with an adequate value, or should prevent
each suspect player to achieve his/her objective.

The value sets can be computed for parity conditions:

\begin{lemma}
  \label{lemma:parity}
  When each $\payoff_A$ is a Boolean payoff function given by a a
  parity condition,\footnote{We assume the reader is familiar with
    parity conditions...} one can compute in exponential space
  $\Value(s)$ for every $s \in S_\Eve \setminus S_\bot$.
\end{lemma}

\begin{proof}
  The condition~\eqref{win} corresponds to a so-called generalized
  (conjunctive) parity condition.  Thanks to~\cite{CHP07}, the winner
  of a generalized conjunctive parity games can be established in
  \coNP. It can in particular be solved in polynomial space.

  Starting from the smallest sets $\calS$, we fill in exponential
  space a table with input a state $s$ of $\calE_\calG$ and a subset
  $\calS'$ of the set $\calS = \suspect(s)$, and with output $1$ or
  $0$, depending on whether \Eve has a winning strategy for
  condition~\eqref{win}. Each computation can be done in exponential
  space. So globally, this can be done in exponential space.
%
%
%
\end{proof}

To conclude, we can state the following result:

\begin{corollary}
  \label{coro:Boolean}
  One can decide in exponential space the (constrained) existence of a
  Nash equilibrium in a game with public signal and publicly visible
  payoff functions associated with parity conditions.
\end{corollary}

\begin{proof}
  We fix two thresholds $(\nu_A)_{A \in \Agt}$ and $(\nu'_A)_{A \in
    \Agt}$ for the payoff.

  Once all the sets $\Value(s)$ for every $s \in S_\Eve \setminus
  S_\bot$ have been computed (in exponential space thanks to
  Lemma~\ref{lemma:parity}), we can do the following: for each payoff
  vector $(p_A)_{A \in\Agt}$ which satisfies $\nu_A \le p_A \le
  \nu'_A$ for every $A \in \Agt$, first mark all states $s$ such that
  $p \in \Value(s)$, and then look for a $\bot$-play which satisfies
  the constraint $\Omega_A$ if $p_A =1$ and $\Omega_A^c$ if $p_A = 0$,
  and only traverses \Adam-states where each deviation to a state $s$
  is marked. This can be done in exponential space as well.  There are
  at most exponentially many payoff vectors $(p_A)_{A \in\Agt}$ such
  that $\nu_A \le p_A \le \nu'_A$ for every $A \in \Agt$, hence this
  concludes the proof of the result. 
\end{proof}

Finally we can state the following result:

\begin{theorem}
  \label{theo:Boolean}
  The constrained existence problem is in
  \EXPSPACE and \EXPTIME-hard for concurrent games with public signal
  and publicly visible Boolean payoff functions associated with parity
  conditions. The lower bound holds even for B\"uchi
  conditions and two players.
\end{theorem}

\begin{proof}
  Only the lower bound remains to be proven. We will use a reduction
  from the halting problem in an alternating linearly-bounded Turing
  machine.

  Let $\calM$ be an alternating linearly-bounded Turing machine.  Let
  $w$ be an input word. On reading $w$, the tape will be bounded by
  $p(|w|)$, where $p$ is a polynom.  W.l.o.g. we rewrite $w$ into a
  length-$p(|w|)$ (completing with $\#$), and assume it is $w$
  itself. We build a two-player turn-based game $\calG_\calM$, which
  will simulate the computation of $\calM$ on $w$. It is constructued
  as follows:
  \begin{itemize}
  \item Initial state $(q_0,0,w(0))$ (belongs to \Adam)
    
    {\small $\calM$ is in cell $0$ and contains the letter $w(0)$}
  \item For all $j \in \{0,\ldots,p(|w|)\}$, $(q_0,0,w(0))
    \xrightarrow{a_{j,w(j)}} (\textit{test},q_0,0,w(0),(j,w(j)))$ and
    $\ell(a_{j,w(j)},\star) = \textit{start}$; State
    $(\textit{test},q_0,0,w(0),(j,w(j)))$ belongs to \Adam

    {\small With that transition, \Adam decides to monitor cell $j$,
      and stores the fact that it contains initially $w(j)$. This is
      invisible to player \Eve, who therefore will not know what \Adam
      is monitoring. Note that we did not check here that $w(j) =
      w(0)$ whenever $j=0$; this will done in the next state.}
  \item $(\textit{test},q,i,\alpha,(j,\beta))
    \xrightarrow{\textit{continue}} (q,i,\alpha,(j,\beta))$ if either
    $i  \ne j$ or $i =j$ and $\alpha=\beta$. \\
    $(\textit{test},q,i,\alpha,(j,\beta))
    \xrightarrow{\textit{failure}} \textsf{fail}$ if $i =j$ but
    $\alpha \ne \beta$

    We set $\ell(\textit{continue},\star) = \textit{continue}$ and
    $\ell(\textit{failure},\star) = \textit{failure}$.
  
  {\small This is a test which can fail whenever \Adam is monitoring
    the current cell, and realizes that what \Eve claims for that cell
  is not correct.}
\item State $(q,i,\alpha,(j,\beta))$ belongs to either \Eve or \Adam,
  depending on whether $q$ is $\vee$ or $\wedge$.

  {\small This is standard alternation between \Eve who tries to find
    the correct next operation and \Adam who tries to find the worst
    operation.}
\item For every state $(q,i,\alpha,(j,\beta))$ (which belongs either
  to \Eve or to \Adam), for every transition $t =
  (q,\alpha,\alpha',q',d)$ such that $0 \le i+d \le p(|w|)$, there is
  a transition $(q,i,\alpha,(j,\beta)) \xrightarrow{t}
  (t,q,i,\alpha,(j,\beta))$; $\ell(t,\star) = t$ and this last state
  belongs to \Eve. If there is no such transition then there is a
  transition $(q,i,\alpha,(j,\beta)) \xrightarrow{\textit{failure}}
  \textsf{fail}$.

  {\small The next move, when it exists, is stored in the state.}
\item Writing $t = (q,\alpha,\alpha',q',d)$, for every $\alpha'' \in
  \Sigma$, $(t,q,i,\alpha,(j,\beta)) \xrightarrow{\alpha''}
  (\textit{test},q',i+d,\alpha'',(j,\beta))$ if $j \ne i$, and
  $(t,q,i,\alpha,(j,\beta)) \xrightarrow{\alpha''}
  (\textit{test},q',i+d,\alpha'',(j,\alpha'))$ if $j =i$.  We set
  $\ell(\alpha'',\star) = \alpha''$.
  
  {\small This is the next move computation. This is where \Eve has to
    be smart, since she has to remember what was in cell $i+d$ (this
    $\alpha''$)! Note that \Adam updates his knowledge about the cell
    he is monitoring...}
\end{itemize}
Note: in the above game, even though the public signal hides some
information, \Adam has full information, since all hidden decisions
have been made by him.

The payoff function for \Eve is given by the reachability condition
``Reach the halting state'', and it is the converse for \Adam. Then
one can then prove the following equivalence between the Turing
machine $\calM$ and the constructed game: $\calM$ halts on $w$ if and
only if there is a Nash equilibrium in $\calG_\calM$ from the initial
state where \Eve wins. 
\end{proof}

\paragraph*{Going further.}  In this section we have only considered
Boolean $\omega$-regular prefix-independent payoff functions. We
easily see that we can twist Lemma~\ref{lemma:botplay} for,
e.g. reachability properties, and obtain a very similar algorithm for
Boolean payoff functions associated with reachability
conditions. We can even twist it further to mix prefix-independent
objectives and reachability objectives. We will not detail that here.

Somewhat more importantly, we could well extend the generic approach
to the so-called \emph{ordered objectives} of~\cite{BBMU15}. They are
somehow finite preference relations (extending payoff functions as
given here) manipulating $\omega$-regular properties: for instance,
given a finite number of $\omega$-regular properties, the ``maximise''
order counts the number of objectives which are satisfied.

Finally let us comment on the public visibility of payoff
functions. While it is important for quantitative objectives like
mean payoff objectives (remember the undecidability result of
Theorem~\ref{theo:MP-undec3}), it is not so important for Boolean
objectives. Indeed we can enrich the epistemic game with extra
information tracking enough information about the past for easily
checking many properties in parallel. We believe we can enrich the
construction 
and provide an algorithm to decide the constrained existence problem
for Boolean $\omega$-regular invisible payoff functions.

\subsection{Mean payoff objectives}

In this section, we assume that each payoff $\payoff_A$ in $\calG$ is
given by a (liminf or limsup) mean payoff function $\MP_A$ that is
publicly visible through the labelling $\ell$: for every $A \in \Agt$,
there is a function $w_A \colon \Sigma \to \bbZ$ such that, if $R =
v_0 \cdot m_1 \cdot v_1 \ldots m_k \cdot v_k \ldots$ is a full play of
$\calG$, then $\payoff_A(R) = \limsup_{k \to \infty} \frac{1}{k}
\sum_{i=1}^k w_A(\ell(m_i,v_i))$ if $\MP_A$ is a limsup mean payoff
function and $\payoff_A(R) = \liminf_{k \to \infty} \frac{1}{k}
\sum_{i=1}^k w_A(\ell(m_i,v_i))$ if $\MP_A$ is a liminf mean payoff
function. Quickly notice that the new payoff functions $\payoff'_A$
used in $\calE_\calG$ are also (liminf or limsup) mean payoff
functions, but since each single step in $\calG$ is mimicked by two
steps in $\calE_\calG$, the corresponding weight functions needs to be
doubled. This will appear in the transformation below.

\medskip First notice that we could have used the generic approach for
this problem as well and apply various results of~\cite{BR14,BR15},
but it is actually easier to follow the approach
of~\cite{brenguier16}, and to transform the winning condition of \Eve
in $\calE_\calG$ directly to a so-called \emph{polyhedron query} in a
multi-dimensional mean-payoff game.

To ease the notations, we write $\Agt = \{A_i \mid 1 \le i \le
|\Agt|\}$. We also denote by $W$ the maximal absolute weight appearing
in $\calG$. We first define several weight functions.  For each $1 \le
i \le |\Agt|$, we define:
\begin{itemize}
\item for every edge $e= \Big(s \xrightarrow{M} (s,M) \Big)$ in
  $\calE_\calG$, we set $w_i(e) = w_{|\Agt|+i}(e) = w_{2|\Agt|+i}(e) =
  0$
\item for every edge $e = \Big((s,M) \xrightarrow{\beta} s')$ in
  $\calE_\calG$ with $s'(\bot)\ne\emptyset$, we set:
  \[
  \left\{\begin{array}{l}
    w_i(e) = 2w_{A_i} (\beta) \\
    w_{|\Agt|+i}(e) = -2w_{A_i}(\beta)  \\
    w_{2|\Agt|+i}(e) = -2w_{A_i}(\beta)
  \end{array}\right.
  \]
\item for every edge $e = \Big((s,M) \xrightarrow{\beta} s')$ in
  $\calE_\calG$ with $s'(\bot) = \emptyset$, we set:
  \[
  \left\{\begin{array}{l}
    w_i(e) = 2W \\
    w_{|\Agt|+i}(e) = 2W \\
    w_{2|\Agt|+i}(e) = \left\{\begin{array}{ll}
        -2w_{A_i}(\beta) & \text{if}\ A_i \in \suspect(s) \\
        2W &  \text{otherwise}
      \end{array}\right.
  \end{array}\right.
  \]
\end{itemize}

For every $i$, we write $\MP_i^{-1}$ for $\underline{\MP_i}$ and
$\MP_i^{+1}$ for $\overline{\MP_i}$, where the weight taken into
account is $w_i$. We also write $\iota(i) =+1$ if $\MP_{A_i}$ is a
limsup mean payoff and $\iota(i)=-1$ if $\MP_{A_i}$ is a liminf
meanpayoff.

We now show that the winning condition of \Eve in $\calE_\calG$ can be
expressed using these new mean payoff functions.

\begin{lemma}
  \label{lemma:polyhedron}
  \Eve has a winning strategy in $\calE_\calG$ from $s_0$ for payoff
  vector $p = (p_{A})_{A \in \Agt}$ if and only if \Eve has a strategy
  $\sigma_\Eve$ from $s_0$ such that
  for every $R \in \out(\sigma_\Eve,s_0)$:
  \[
  \left\{\begin{array}{l}
      \MP^{\iota(i)}_i (R) \ge p_{A_i} \\[.2cm]
      \MP^{-\iota(i)}_{|\Agt|+i} (R) \ge -p_{A_i} \\[.2cm]
      \MP^{-\iota(i)}_{2|\Agt|+i} (R) \ge -p_{A_i}
    \end{array}\right.
  \]
\end{lemma}

\begin{proof}
  Pick a winning strategy $\sigma_\Eve$ for \Eve in $\calE_\calG$ from
  $s_0$, for payoff $p = (p_A)_{A \in \Agt}$. In particular,
  $\payoff'(\out_\bot(\sigma_\Eve,s_0)) = p$, and for all $R \in
  \out(\sigma_\Eve,s_0)$, for every $A \in \suspect(R)$,
  $\payoff'_A(R) \le p_A$. We will show that $p$ satisfies the
  constraints in the statement.
  
  Consider first $R = \out_\bot(\sigma_\Eve,s_0)$. For every $A_i
  \in \Agt$, $\payoff'_{A_i}(R) = p_{A_i}$. By definition of the
  weights in the $\bot$-part of the game, we get $\payoff'_{A_i}(R)
  = \MP^{\iota(i)}_i(R) = - \MP^{-\iota(i)}_{|\Agt|+i}(R) =
  -\MP^{-\iota(i)}_{2|\Agt|+i} (R)$; hence $\MP^{\iota(i)}_i (R)
  \ge p_{A_i}$, $\MP^{-\iota(i)}_{|\Agt|+i} (R) \ge - p_{A_i}$ and
  $\MP^{-\iota(i)}_{2|\Agt|+i} (R) \ge - p_{A_i}$.
  
  Consider now $R \in \out(\sigma_\Eve,s_0) \setminus
  \{\out_\bot(\sigma_\Eve,s_0)\}$. For every $A_i \in \suspect(R)$,
  $\payoff'_{A_i}(R) \le p_{A_i}$; there is no constraint on
  $\payoff'_{A_i}(R)$ when $A_i \notin \suspect(R)$. By
  definition of the weights in the game:
  \begin{itemize}
  \item $\MP_i^{\iota(i)}(R) = \MP_i^{-\iota(i)}(R) = W$, hence
    $\MP_i^{\iota(i)}(R) \ge p_{A_i}$ (since $p_{A_i}$ is a
    mean payoff value in $\calG$, hence it is bounded by $W$);
  \item $\MP_{|\Agt|+i}^{\iota(i)}(R) =
    \MP_{|\Agt|+i}^{-\iota(i)}(R) = W$, hence
    $\MP_{|\Agt|+i}^{-\iota(i)}(R) \ge -p_{A_i}$;
  \item if $A_i \notin \suspect(R)$,
    $\MP_{2|\Agt|+i}^{\iota(i)}(R) =
    \MP_{2|\Agt|+i}^{-\iota(i)}(R) = W$, hence
    $\MP_{2|\Agt|+i}^{-\iota(i)}(R) \ge -p_{A_i}$;
  \item if $A_i \in \suspect(R)$, $\payoff'_{A_i}(R) =
    -\MP^{-\iota(i)}_{2|\Agt|+i} (R) \le p_{A_i}$; hence
    $\MP^{-\iota(i)}_{2|\Agt|+i} (R) \ge -p_{A_i}$.
  \end{itemize}
  It implies that $\sigma_\Eve$ is winning for the new winning
  condition.

  \medskip Conversely assume that $\sigma_\Eve$ is winning for the new
  winning condition. Then consider $R =
  \out_\bot(\sigma_\Eve,s_0)$. It holds that $\payoff'_{A_i}(R) =
  \MP^{\iota(i)}_i(R) = -\MP^{-\iota(i)}_{|\Agt|+i}(R)$, hence
  we get $\payoff'_{A_i} (R) = p_{A_i}$. Pick now $R \in
  \out(\sigma_\Eve,s_0) \setminus \{\out_\bot(\sigma_\Eve,s_0)\}$, and
  $A_i \in \suspect(R)$. Then, $\payoff'_{A_i}(R) =
  -\MP^{-\iota(i)}_{2|\Agt|+i}(R)$, hence $\payoff'_{A_i}(R) \le
  p_{A_i}$. We conclude that $\sigma_\Eve$ is a winning strategy for
  the original objective. 
\end{proof}

For every $1 \le i \le |\Agt|$, $\widetilde{\MP}_i$ is set to
$\MP_i^{\iota(i)}$, $\widetilde{\MP}_{|\Agt|+i}$ is set to
$\MP^{-\iota(i)}_{|\Agt|+i}$ and $\widetilde{\MP}_{2|\Agt|+i}$ is set
to $\MP^{-\iota(i)}_{2|\Agt|+i}$. Gathering
Theorem~\ref{theo:correction} and Lemma~\ref{lemma:polyhedron}, we can
characterize the constrained existence problem for mean payoff
functions as follows.

\begin{corollary}
  Let $\nu,\nu' \in \overline{\bbR}^\Agt$. There is a Nash equilibrium
  in $\calG$ from $v_0$ with payoff $p$ such that $\nu \le p \le \nu'$
  if and only if there is $u=(u_i)_{1 \le i \le 3|\Agt|} \in
  \overline{\bbR}^{3|\Agt|}$ such that:
  \begin{enumerate}[(1)]
  \item for every $1 \le i \le |\Agt|$, $u_i = -u_{|\Agt|+i} =
    -u_{2|\Agt|+i}$ and $\nu_i \le u_i \le \nu'_i$;
  \item \Eve has a strategy from $s_0$ in $\calE_\calG$ to ensure
    $\widetilde{\MP}_i \ge u_i$ for every $1 \le i \le 3|\Agt|$.
  \end{enumerate}
\end{corollary}

This above problem is known as the polyhedron problem
in~\cite{BR15}. It is shown in this paper that if there is a solution,
there is one solution whose encoding has size polynomial in the number
of dimensions (here $3|\Agt|$), the encoding of the polyhedron (here
$4|\Agt|$), the encoding of the maximal weight (here, $\log_2(W)$) and
the encoding of the number of states of the game (here,
$O(\log_2(2^{|\Agt| \cdot |V|} \cdot |\Act|^{|\Agt|^2 \cdot |V|}))$,
that is $ = O(|\Act| \cdot |\Agt|^2 \cdot |V|)$).  Hence we can guess
in polynomial time a possible encoding for such a solution $u$, check
in polynomial time that it belongs to the polyhedron, and use an
oracle to decide whether \Eve is winning in this game. We can apply
results of~\cite{VCD+15} (which show that multi-dimensional
mean-payoff games can be solved in \coNP) to infer that this can be
done in \coNEXPTIME (since $\calE_\calG$ has exponential-size and new
weights have polynomial encodings). Globally, everything can therefore
be done in exponential space.  We therefore deduce the following
result.

\begin{theorem}
  \label{theo:MP}
  The constrained existence problem is in
  \NP$^{\!\!\text{\textsf{\upshape{NEXPTIME}}}}$ (hence in \EXPSPACE) and
  \EXPTIME-hard for concurrent games with public signal and publicly
  visible mean payoff functions.
\end{theorem}


The lower bound uses the result for Boolean objectives stated in
Theorem~\ref{theo:Boolean}.

\section{Conclusion}

In this paper, we have studied concurrent games with imperfect
monitoring modelled using signals. We have given some undecidability
results, even in the case of public signals, when the payoff functions
are not publicly visible. We have then proposed a construction to
capture single-player deviations in games with public signals, and
reduced the search of Nash equilibria to the synthesis of winning
strategies in a two-player turn-based games (with a rather complex
winning condition though). We have applied this general framework to
two classes of payoff functions, and obtained decidability
results.

As further work we wish to understand better if there could be richer
communication patterns which would allow representable knowledge
structures for Nash equilibria and thereby the synthesis of Nash
equilibria under imperfect monitoring. A source of inspiration for
further work will be~\cite{RT98}.
 
\bigskip\noindent \textbf{Acknowledgement.} We would like to thank
Sylvain Schmitz for his help in better understanding complexity
classes between \EXPTIME and \EXPSPACE.


\newcommand{\etalchar}[1]{$^{#1}$}

\appendix

\section{Why players should distinguish the actions they played}
\label{app:discuss}

In this paper, to model the perfect-recall hypothesis, we define the
undistinguishability relation $\sim_A$ of player $A$ using projection
$\pi_A$, which includes private actions of player $A$ and the
signal. Standardly though (see e.g.~\cite{MW05,BKP11,DR11}), the
undistinguishability relation is defined using solely the signal
(hence forgetting private actions). While this does not affect the
recall for distributed synthesis (existence of a strategy for the
grand coalition), nor for Nash equilibria (existence of a strategy
profile as well), we believe this is important for more complex
interactions between players (for instance for multi-agent logics like
in~\cite{BMM+17}).

We will argue this point using subgame-perfect equilibria. Those are
strategy profiles for which, for any history (generated or not by the
main profile), the strategy profile after that history is a Nash
equilibrium. Consider the game represented below, with Boolean
objectives as indicated, and where vertices $v_2$ and $v_3$ cannot be
distinguished by any of the two players

  
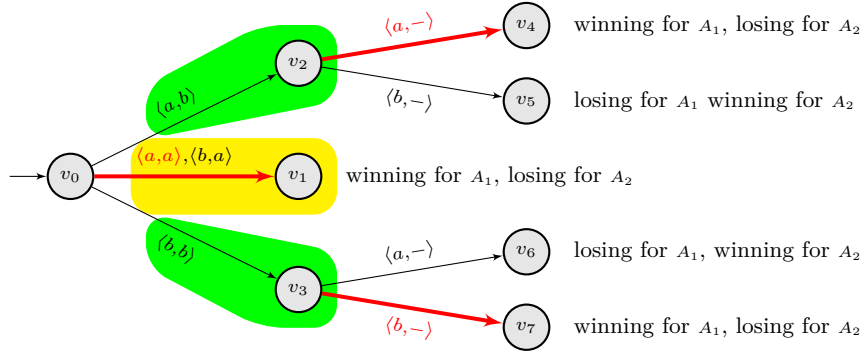
\begin{center}
    \begin{tikzpicture}
      \everymath{\scriptstyle}
      \tikzset{noeud/.style={circle,draw=black,thick,fill=black!10,minimum
          height=6mm,inner sep=0pt}}

      \filldraw [green,rounded corners=3mm] (1,.5) -- (3.5,1) --
      (3.5,2) -- (2.5,2) -- (1,1.2) --cycle;
          
      \filldraw [green,rounded corners=3mm] (1,-.5) -- (3.5,-1) --
      (3.5,-2) -- (2.5,-2) -- (1,-1.2) --cycle;

      \filldraw [yellow,rounded corners=3mm] (.8,-.5) -- (3.5,-.5) --
      (3.5,.5) -- (.8,.5) --cycle;

      \draw (0,0) node [noeud] (A) {$v_0$};
      \draw (3,0) node [noeud] (B) {$v_1$} node [right=.5cm]
      {{\footnotesize winning for $A_1$, losing for $A_2$}};
      \draw (3,1.5) node [noeud] (C) {$v_2$};
      \draw (3,-1.5) node [noeud] (D) {$v_3$};
      \draw (6,2) node [noeud] (C1) {$v_4$} node [right=.5cm]
      {{\footnotesize winning for $A_1$, losing for $A_2$}};
      \draw (6,1) node [noeud] (C2) {$v_5$} node [right=.5cm]
      {{\footnotesize losing for $A_1$ winning for $A_2$}};
      \draw (6,-1) node [noeud] (D1) {$v_6$} node [right=.5cm]
      {{\footnotesize losing for $A_1$, winning for $A_2$}};
      \draw (6,-2) node [noeud] (D2) {$v_7$} node [right=.5cm]
      {{\footnotesize winning for $A_1$, losing for $A_2$}};

      \draw [latex'-] (A) -- +(-.8,0);

      \draw [-latex',red,line width=1.5pt] (A) -- (B) node [midway,above]
      {$\textcolor{red}{\tuple{a,a}}\textcolor{black}{,\tuple{b,a}}$};
      \draw [-latex'] (A) -- (C) node [midway,above,sloped]
      {$\tuple{a,b}$};
      \draw [-latex',red,line width=1.5pt] (C) -- (C1) node [midway,above,sloped]
      {$\tuple{a,-}$};
      \draw [-latex'] (C) -- (C2) node [midway,below,sloped]
      {$\tuple{b,-}$};
      \draw [-latex'] (A) -- (D) node [midway,below,sloped]
      {$\tuple{b,b}$};
      \draw [-latex'] (D) -- (D1) node [midway,above,sloped]
      {$\tuple{a,-}$};
      \draw [-latex',red,line width=1.5pt] (D) -- (D2) node [midway,below,sloped]
      {$\tuple{b,-}$};
      
    \end{tikzpicture}
  \end{center}
  The strategy profile $\sigma_\Agt$ represented in bold red, and
  defined by $\sigma_\Agt(v_0) = \tuple{a,a}$, $\sigma_\Agt(v_0 \cdot
  \tuple{a,b} \cdot v_2) = \tuple{a,-}$ and $\sigma_\Agt(v_0 \cdot
  \tuple{b,b} \cdot v_3) = \tuple{b,-}$, is a subgame-perfect
  equilibrium in our framework. Indeed, player $A_1$ distinguishes the
  two histories $v_0 \cdot \tuple{a,b} \cdot v_2$ and $v_0 \cdot
  \tuple{b,b} \cdot v_3$, since she can only be in $v_2$ (resp. $v_3$)
  if she played $a$ (resp. $b$). On the other hand, this is not a
  proper profile in the standard framework since player $A_1$ is not
  supposed to distinguish these histories. We believe this shows that
  standard assumptions do not properly model perfect recall.

\end{document}

%% file: figure.tex

  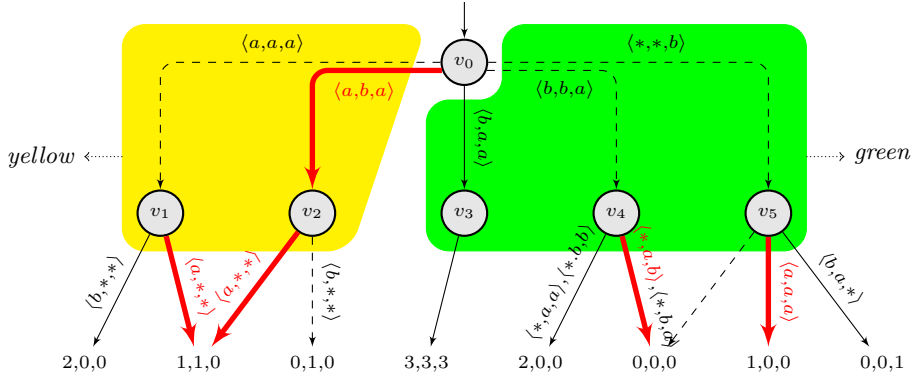
\begin{figure}[t]
    \begin{center}
    \begin{tikzpicture}[rounded corners=2mm]
      \everymath{\scriptstyle}
      

      \tikzset{noeud/.style={circle,draw=black,thick,fill=black!10,minimum
          height=6mm,inner sep=0pt}}
      
      \filldraw [green,rounded corners=3mm] (-.5,-.5) -- (-.5,-2.5) --
      (4.5,-2.5) -- (4.5,.5) --(.5,.5) -- (.5,-.5) --cycle;

      \filldraw [yellow,rounded corners=3mm] (-.5,.5) -- (-1.5,-2.5) --
      (-4.5,-2.5) -- (-4.5,.5) --cycle;

      \draw [->,densely dotted] (-4.5,-1.25) -- +(-.5,0) node
      [left]{{\small \textit{yellow}}};
      \draw [->,densely dotted] (4.5,-1.25) -- +(.5,0) node
      [right]{{\small \textit{green}}};
      
      \draw (0,0) node [noeud] (v) {$v_0$};
      \draw (-4,-2) node [noeud] (v1) {$v_1$};
      \draw (-2,-2) node [noeud] (v2) {$v_2$};
      \draw (0,-2) node [noeud] (v3) {$v_3$};
      \draw (2,-2) node [noeud] (v4) {$v_4$};
      \draw (4,-2) node [noeud] (v5) {$v_5$};
      \draw (-5,-4) node (200) {$2,0,0$};
      \draw (-3.5,-4) node (110) {$1,1,0$};
      \draw (-2,-4) node (010) {$0,1,0$};
      \draw (-.5,-4) node (333) {$3,3,3$};
      \draw (1,-4) node (200bis) {$2,0,0$};
      \draw (2.5,-4) node (000) {$0,0,0$};
      \draw (4,-4) node (100) {$1,0,0$};
      \draw (5.5,-4) node (001) {$0,0,1$};
      \draw [latex'-] (v) -- +(0,.8);
      \draw [-latex',dashed] (v.180) -| (v1) node [pos=.3,sloped,above] {$\tuple{a,a,a}$};
      \draw [-latex',red,line width=2pt] (v.-160) -| (v2) node [pos=.3,sloped,below] {$\tuple{a,b,a}$};
      \draw [-latex'] (v) -- (v3) node [pos=.5,sloped,above] {$\tuple{b,a,a}$};
      \draw [-latex',dashed] (v.-20) -| (v4) node [pos=.3,sloped,below] {$\tuple{b,b,a}$};
      \draw [-latex',dashed] (v) -| (v5) node [pos=.3,sloped,above] {$\tuple{*,*,b}$};
      \draw [-latex',red,line width=2pt] (v2) -- (110) node [midway,sloped,above] {$\tuple{a,*,*}$};
      \draw [-latex',dashed] (v2) -- (010) node [midway,sloped,above] {$\tuple{b,*,*}$};
      \draw [-latex'] (v1) -- (200) node [midway,sloped,above] {$\tuple{b,*,*}$};
      \draw [-latex',red,line width=2pt] (v1) -- (110) node [midway,sloped,above] {$\tuple{a,*,*}$};
      \draw [-latex'] (v3) -- (333) node [midway,sloped,above]  {};
      \draw [-latex'] (v4) -- (200bis) node [midway,sloped,above] {$\tuple{*,a,a},\tuple{*.b,b}$};
      \draw [-latex',red,line width=2pt] (v4) -- (000) node [midway,sloped,above] {$\tuple{*,a,b}\textcolor{black}{,\tuple{*,b,a}}$};
      \draw [-latex',dashed] (v5) -- (000);
      \draw [-latex',red,line width=2pt] (v5) -- (100) node [midway,sloped,above] {$\tuple{a,a,a}$};
      \draw [-latex'] (v5) -- (001) node [midway,sloped,above] {$\tuple{b,a,*}$};

    \end{tikzpicture}
    \end{center}
    \caption{An example of a concurrent game with public signal
      (yellow and green: public signal). Edges in red and bold are
      part of the strategy profile. Dashed edges are the possible
      deviations. One can notice that none of the deviations is
      profitable to the deviator, hence the strategy profile is a Nash
      equilibrium. Convention in the drawing: edges with no label are
      for complementary labels (for instance the edge from $v_5$ to
      $0,0,0$ is labelled by all $\tuple{a_1,a_2,a_3}$ not in the set
      $\{\tuple{a,a,a},\tuple{b,a,a},\tuple{b,a,b}\}$)}
    \label{fig:ex}
  \end{figure}

%% file: figure-epistemic.tex
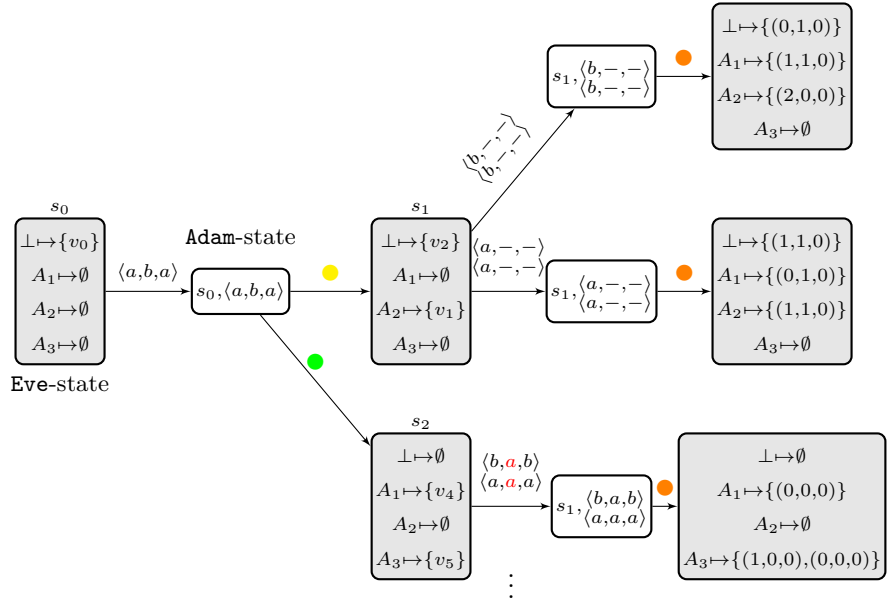
\begin{figure}[t]
  \begin{center}
    \begin{tikzpicture}[scale=.95]
      \everymath{\scriptstyle}
      \tikzset{Eve/.style={draw=black,thick,fill=black!10,minimum
          height=6mm,inner sep=1.5pt,rounded corners=3pt}}   
      \tikzset{Adam/.style={draw=black,thick,minimum
          height=6mm,inner sep=1.5pt,rounded corners=3pt}}


      \draw (0,0) node [Eve] (s0) {$\begin{array}{@{}c@{}} \bot 
          \mapsto  \{v_0\} \\ A_1  \mapsto  \emptyset \\ A_2 
          \mapsto  \emptyset \\
          A_3  \mapsto  \emptyset \end{array}$} node [above=.9cm] {$s_0$}
      node [below=1cm] {{\small \Eve-state}}; 
      \draw (2.5,0) node [Adam] (s0m0) {$s_0,\tuple{a,b,a}$} node
      [above=.5cm] {{\small \Adam-state}};

      \draw (5,0) node [Eve] (s1) {$\begin{array}{@{}c@{}} \bot
          \mapsto \{v_2\} \\ A_1 \mapsto \emptyset \\ A_2
          \mapsto \{v_1\} \\
          A_3 \mapsto \emptyset \end{array}$} node [above=.9cm]
      {$s_1$};

     \draw (5,-3) node [Eve] (s2) {$\begin{array}{@{}c@{}}
         \bot \mapsto \emptyset \\
          A_1 \mapsto \{v_4\} \\ A_2 \mapsto \emptyset \\
          A_3 \mapsto \{v_5\} \end{array}$} node [above=.9cm] {$s_2$};

      \draw (6.25,-4) node {$\vdots$};

      \draw (7.5,0) node [Adam] (s1m1) {$s_1,\begin{array}{@{}c@{}} \tuple{a,-,-} \\[-.2cm] \tuple{a,-,-} \end{array}$};
      \draw (10,0) node [Eve] (out1) {$\begin{array}{@{}c@{}}
          \bot \mapsto \{(1,1,0)\} \\ A_1 \mapsto \{(0,1,0)\} \\ A_2
          \mapsto \{(1,1,0)\}
          \\ A_3 \mapsto \emptyset \end{array}$};
      \draw (7.5,-3) node [Adam] (s1m11) {$s_1,\begin{array}{@{}c@{}}
          \tuple{b,a,b} \\[-.2cm] \tuple{a,a,a} \end{array}$}; 
      \draw (10,-3) node [Eve] (out2) {$\begin{array}{@{}c@{}}
          \bot \mapsto \emptyset \\ A_1 \mapsto \{(0,0,0)\} \\ A_2
          \mapsto \emptyset
          \\ A_3 \mapsto \{(1,0,0),(0,0,0)\} \end{array}$};
      \draw (7.5,3) node [Adam] (s1m12) {$s_1,\begin{array}{@{}c@{}} \tuple{b,-,-} \\[-.2cm] \tuple{b,-,-} \end{array}$};
      \draw (10,3) node [Eve] (out3) {$\begin{array}{@{}c@{}}
          \bot \mapsto \{(0,1,0)\} \\ A_1 \mapsto \{(1,1,0)\} \\ A_2
          \mapsto \{(2,0,0)\} \\
          A_3 \mapsto \emptyset  \end{array}$};

      \draw [-latex'] (s0) -- (s0m0) node
      [midway,above]{$\tuple{a,b,a}$};
      \draw [-latex'] (s0m0) -- (s1) node
      [midway,above]{\textcolor{yellow}{\scalebox{2}{$\bullet$}}};
      \draw [-latex'] (s0m0) -- (s2) node
      [midway,above]{\textcolor{green}{\scalebox{2}{$\bullet$}}};
      \draw [-latex'] (s1) -- (s1m1) node
      [midway,above]{$\begin{array}{@{}c@{}} \tuple{a,-,-} \\[-.2cm]
          \tuple{a,-,-} \end{array}$};
      \draw [-latex'] (s1m1) -- (out1) node [midway,above]
      {\textcolor{orange}{\scalebox{2}{$\bullet$}}};
      \draw [-latex'] (s2) -- (s1m11) node
      [midway,above]{$\begin{array}{@{}c@{}} \tuple{b,\textcolor{red}{a},b} \\[-.2cm]
          \tuple{a,\textcolor{red}{a},a} \end{array}$};
      \draw [-latex'] (s1m11) -- (out2)  node [midway,above]
      {\textcolor{orange}{\scalebox{2}{$\bullet$}}}; 
      \draw [-latex'] (s1) -- (s1m12) node [midway,above,sloped]
      {$\begin{array}{@{}c@{}} \tuple{b,-,-} \\[-.2cm]
          \tuple{b,-,-} \end{array}$}; 
      \draw [-latex'] (s1m12) -- (out3) node [midway,above]
      {\textcolor{orange}{\scalebox{2}{$\bullet$}}}; 
  \end{tikzpicture}
\end{center}
\caption{Part of the epistemic game corresponding to the game of
  Fig.~\ref{fig:ex}. For clarity, symbol $-$ is for any choice $a$ or
  $b$ (the precise choice is meaningless).} \label{fig:epistemic}
\end{figure}